\documentclass[10pt,fleqn,leqno]{article}
\usepackage[english]{babel}
\usepackage{amsfonts,amsmath,amssymb, amsthm,dsfont,enumitem,epsfig,graphicx,hyperref,indentfirst,multirow,rotating,subfigure,url,xcolor}
\hypersetup{colorlinks=true, allcolors=blue} 

\makeatletter
\@addtoreset{equation}{section}
\long\def\@makecaption#1#2{%
  \vskip\abovecaptionskip
  \sbox\@tempboxa{#1: #2}%
  \ifdim \wd\@tempboxa >\hsize
    #1: #2\par
  \else
    \global \@minipagefalse
    \hb@xt@\hsize{\hfil\box\@tempboxa\hfil}%
  \fi
  \vskip\belowcaptionskip}
\makeatother
\theoremstyle{plain}

\theoremstyle{definition}

\newtheorem{prop}{Proposition}[section]

\newcommand{\0}{\bm{0}}

\newcommand{\be}{\begin{eqnarray}}
\newcommand{\ee}{\end{eqnarray}}
\newcommand{\bq}{\begin{eqnarray*}}
\newcommand{\eq}{\end{eqnarray*}}

\newcommand{\bm}[1]{{\boldsymbol{#1}}}
\newcommand{\bA}{{\bm{\rm{A}}}}
\newcommand{\balpha}{{\bm{\alpha}}}
\newcommand{\bg}{{\bm{g}}}
\newcommand{\bh}{{\bm{h}}}
\newcommand{\bH}{{\bm{\rm{H}}}}
\newcommand{\bI}{{\bm{\rm{I}}}}
\newcommand{\blambda}{{\bm{\lambda}}}
\newcommand{\bOmega}{{\bm{\Omega}}}
\newcommand{\bQ}{{\bm{\rm{Q}}}}
\newcommand{\brho}{{\bm{\rho}}}
\newcommand{\bs}{{\bm{s}}}
\newcommand{\bSig}{{\bm{\Sigma}}}
\newcommand{\bt}{{\bm{t}}}
\newcommand{\btau}{{\bm{\tau}}}
\newcommand{\btheta}{{\bm{\theta}}}
\newcommand{\bu}{{\bm{u}}}
\newcommand{\bv}{{\bm{v}}}
\newcommand{\bvtheta}{{\bm{\vartheta}}}
\newcommand{\bX}{{\bm{X}}}
\newcommand{\bxi}{{\bm{\xi}}}

\newcommand{\bY}{{\bm{Y}}}
\newcommand{\bW}{{\bm{W}}}
\newcommand{\bZ}{{\bm{Z}}}
\newcommand{\E}{{\mathbb E}}
\newcommand{\rS}{{\rm S}}

\newcommand{\R}{\mathbb{R}}
\newcommand{\Sp}{\mathcal{S}_p}
\newcommand{\X}{{\cal X}}

\newcounter{sim}
\newenvironment{simul}[1][]{\refstepcounter{sim}\par\medskip
   \noindent \textbf{Simulation~\thesim. #1}}{\smallskip}

\begin{document}
\frenchspacing
\begin{center}
{ \LARGE\sc Goodness-of-fit tests for multivariate skewed distributions based on the characteristic function \\
} \vspace*{0.5cm} {\large\sc } 
{\large\sc Maicon J. Karling}$^{a}$\footnote{Corresponding author: \href{mailto:maicon.karling@kaust.edu.sa}{maicon.karling@kaust.edu.sa}}, 
{\large\sc Marc G. Genton}$^{a}$,
{\large\sc Simos G. Meintanis}$^{b,c}$\footnote{On sabbatical leave from the University of Athens.}  \\ \vspace*{0.5cm}
{\it $^{a}$Statistics Program, King Abdullah University of Science and Technology, Thuwal, Saudi Arabia} \\
{\it $^{b}$Department of Economics, National and Kapodistrian University
of Athens, Athens, Greece} \\
{\it $^{c}$Pure and Applied Analytics, North-West University, Potchefstroom, South Africa} \\
\vspace{.3cm}
\today
\end{center}


\noindent {\small {\bf Abstract.}
We employ a general Monte Carlo method to test composite hypotheses of goodness-of-fit for several popular multivariate models that can accommodate both asymmetry and heavy tails. Specifically, we consider weighted L2-type tests based on a discrepancy measure involving the distance between empirical characteristic functions and thus avoid the need for employing corresponding population quantities which may be unknown or complicated to work with. The only requirements of our tests are that we should be able to draw samples from the distribution under test and possess a reasonable method of estimation of the unknown distributional parameters. Monte Carlo studies are conducted to investigate the performance of the test criteria in finite samples for several families of skewed distributions. Real-data examples are also included to illustrate our method.}
\vspace{.3cm}


\noindent {\small {\it AMS 2020 classification numbers:} 
62F03, 
62H12, 
62H15. 
}


\section{Introduction}
Since the late 1980s, L2-type tests for goodness-of-fit based on the characteristic function (CF) have witnessed increasing popularity. The main reason is that the CF uniquely determines the underlying distribution and that it may be consistently estimated by the empirical CF. For multivariate distributions, there is the extra advantage that multivariate CFs and empirical CFs are well-defined and smooth, unlike the cumulative distribution function and its empirical counterpart, and thus it is easier to work with, even when the population distribution function is known. 

Not surprisingly, testing for multivariate normality occupies a prominent place in this setting (see, e.g., Chen and Genton \cite{CG22}, Ebner et al. \cite{EHS21}, Henze \cite{H02}, Henze et al. \cite{HJM19}, Henze and Wagner \cite{HW97}, and Pude\l{}ko \cite{P05}) as a wide range of procedures is available, including CF-based tests. Outside the multivariate Gaussian context, however, the range of CF-based goodness-of-fit procedures is limited to only a handful of distributions, most of them belonging to the elliptical class (see, e.g., Fragiadakis and Meintanis \cite{FM11}, Meintanis et al. \cite{MNT15}, and Sz\'ekely and Rizzo \cite{SR13}). One of the main reasons for this lack of available procedures is that CFs and empirical CFs, despite being smooth, are often required to be numerically integrated in the L2 setting, a task that may be problematic in higher dimensions, let alone the fact that the analytic form of the population CFs may be altogether unknown for most multivariate distributions under test. 

Recently Chen et al. \cite{CJMZ22} proposed a procedure that is based on a Monte Carlo approximation of the CF under test, thereby avoiding the use of corresponding population quantities. However, the elliptical families considered by Chen et al. \cite{CJMZ22}, as important as they may be, render a range of shapes that limit their potential application, given the fact that asymmetry, in addition to excess kurtosis, is typically expected in real data analysis from Economics, Finance, and most other disciplines. 

In this paper, we follow the approach suggested by Chen et al. \cite{CJMZ22}, but, at the same time, abandon the context of ellipticity adopted therein towards more general shapes. Specifically, we consider goodness-of-fit tests for certain popular families of multivariate skewed distributions. In this connection, an extra element that needs to be addressed in implementing the tests compared to Chen et al. \cite{CJMZ22} is that, unlike the parameters-free tests proposed in that paper, in the current setting the presence of shape parameters necessitates an additional re-sampling cycle to replicate the empirical distribution of the test statistic for a given parameter configuration. In doing so we take advantage of the canonical form of the distributional family under test, whenever available. In Section \ref{sec2}, we revisit some of the main ideas of the work by Chen et al. \cite{CJMZ22} and the background for our tests shall be provided.

The remainder of this work unfolds as follows. Section \ref{sec3} specifies the actual implementation of the new test procedure using bootstrap re-sampling. In Section \ref{sec4} we introduce and provide a short review of the collection of families that shall be used for the simulations and goodness-of-fit tests and study their respective canonical forms. An extensive Monte Carlo study is presented to illustrate the finite-sample properties of the tests in Section \ref{sec5}. The paper concludes with several real-data applications in Section \ref{sec6} and a discussion of the overall results in Section \ref{sec7}. An online Supplement contains some extra Monte Carlo results.  

\section{Characteristic function-based tests} \label{sec2}

Let $\bX \in \mathbb R^p$ ($p\geq 1$) be a random vector with an absolutely continuous distribution function $F_\bX$.  We are interested in  the composite goodness-of-fit testing problem represented by the null hypothesis
\be \label{null1} 
 {\cal{H}}_{0}: \mbox{the law of} \ \bX   \in {\cal{F}}_{\bvtheta}, \quad \mbox{for some} \ \bvtheta \in \Theta,
\ee
where ${\cal {F}}_{\bvtheta}=\{ F_{\bvtheta}, \, \bvtheta \in \Theta\}$ denotes a specific parametric family of distributions admitting a parameterization in terms of the parameter vector $\bvtheta$. The corresponding parameter space $\Theta$ will be taken as an open subset of $\mathbb R^q$ ($q\geq 1$). 
Given the uniqueness of CFs, we may equivalently state the null hypothesis in \eqref{null1} as
\be \label{null2} 
 \varphi_\bX(\bt)=\varphi_{\bvtheta}(\bt), \quad \forall \, \bt\in \mathbb R^p, \ \mbox{for some} \ \bvtheta \in \Theta,
\ee
where $\varphi_\bX(\bt)=\mathbb E(e^{{\rm{i}} \bt^\top  \bX})$ denotes the CF of $\bX$ and $\varphi_{\bvtheta}(\bt)$ corresponds to the CF of some random vector in the family ${\cal {F}}_{\bvtheta}$. Here ${\rm{i}}=\sqrt{-1}$ and $\top$ means transposition of vectors and matrices. 

A CF-based statistic for goodness-of-fit is typically formulated in terms  of  $\|\varphi_{n}-\varphi_{\widehat{\bvtheta}_n} \|_w^2$, where 
\begin{equation} \label{Hilb}
\|f- g \|^2_w:=\int_{\mathbb R^p} |f(\bt)-g(\bt)|^2 \, w(\bt) \, {\rm{d}}\bt
\end{equation}
is an L2-type weighted distance between the pair of complex-valued functions $(f,g)$,   
\be \label{ECF} 
\varphi_{n}(\bt)=\frac{1}{n} \sum_{j=1}^n e^{{\rm{i}} \bt^\top \bX_j}  
\ee
is the empirical CF computed from a collection $(\bX_1,\ldots,\bX_n)$ of independent and identically distributed (i.i.d.) copies of $\bX$, and $\varphi_{\widehat{\bvtheta}_n}$ is the CF corresponding to the null hypothesis ${\cal{H}}_0$ with the parameter $\bvtheta$ replaced by an estimator  $\widehat{\bvtheta}_n:=\widehat{\bvtheta}_n(\bX_1,\ldots,\bX_n)$. The weight function $w>0$ will be further specified below. 

There exist cases though of distributions, some of which will be considered herein, for which the null CF $\varphi_{\bvtheta}(\cdot)$ is either completely unknown or too complicated to work with. In such cases we suggest formulating a test statistic analogously but without direct reference to the CF of the distribution under test. Specifically, and in line with Chen et al. \cite{CJMZ22}, we suggest the test statistic 
\be \label{ts}
 T^{(w)}_{n,m} = \|\varphi_{n}-\widehat {\varphi}_{0,m} \|_w^2,
\ee 
where $\varphi_{n}(\cdot)$ is as in \eqref{ECF}, while 
\be 
\widehat {\varphi}_{0,m}(\bt)= \frac{1}{m} \sum_{j=1}^m e^{{\rm{i}} \bt^\top \bX_{0,j}}  
\ee
is an empirical CF computed from a sample $(\bX_{0,1},\ldots,\bX_{0,m})$ which is drawn from $F_{\widehat{\bvtheta}_n}$, i.e., from a sample of size $m$ ($m\geq n$) taken from the distribution under test with parameter estimated by a consistent estimator $\widehat{\bvtheta}_n:=\widehat{\bvtheta}_n(\bX_1,\ldots,\bX_n)$. In other words, $\widehat {\varphi}_{0,m}$ is a Monte Carlo approximation of the null CF $\varphi_\bvtheta(\cdot)$. Rejection is for large values of $T^{(w)}_{n,m}$.

A clear advantage of using the test statistic $T^{(w)}_{n,m}$ is its computational simplicity. To see this, write $|\cdot|$ for the modulus of a complex number and, thereafter, by using standard algebra, we obtain
\begin{eqnarray*} \nonumber
|\varphi_{n}(\bt)-\widehat {\varphi}_{0,m}(\bt)|^2&=&\frac{1}{n^2} \sum_{j,k=1}^n \cos \bt^\top (\bX_j-\bX_k)+\frac{1}{m^2} \sum_{j,k=1}^m \cos \bt^\top (\bX_{0,j}-\bX_{0,k})\\ \label{phi2}  &&- \frac{2}{nm} \sum_{j=1}^n \sum_{k=1}^m \cos \bt^\top (\bX_j-\bX_{0,k}).
\end{eqnarray*}
Now suppose that the weight function $w(\cdot)$, figuring in \eqref{Hilb} and \eqref{ts}, is chosen as the density of a random vector $\bW \in \mathbb R^p$
following a certain spherical distribution. Then it is well known that the CF of $\bW$ simplifies to $\varphi_{\bW}(\bt)=\mathbb E(\cos \bt^\top \bW)$ and it is eventually given by $\Psi(\|\bt\|^2)$, where $\Psi(\cdot)$ is called the ``kernel" associated with $\bW$ and $\| \cdot \|$ stands for the standard Euclidean norm in $\mathbb R^p$ (see Fang et al. \cite{FKN90}). By using the last equation and such a weight function $w(\cdot)$ in \eqref{ts}, we end up with the test statistic 
\begin{eqnarray} \label{psi} 
T^{(\Psi)}_{n,m}&=&\frac{1}{n^2} \sum_{j,k=1}^n \Psi(\|\bX_j-\bX_k\|^2) 
+\frac{1}{m^2} \sum_{j,k=1}^m \Psi( \|\bX_{0,j}-\bX_{0,k}\|^2)\\  \nonumber &&- \frac{2}{nm} \sum_{j=1}^n \sum_{k=1}^m \Psi(\|\bX_j-\bX_{0,k}\|^2),
\end{eqnarray}
where we have made the dependence of the test statistic on the kernel $\Psi$ explicit. Provided that the kernel $\Psi(\cdot)$ is simple enough, \eqref{psi} can be readily computed. Some prominent examples of simple kernels at our disposal are:
\begin{itemize}
    \item the standard normal kernel $\Psi(\xi)=e^{-\xi/2}$;
    
    \item the kernel $\Psi(\xi)=e^{- \xi^{b/2}}$, $b \in (0,2)$, that originates from the stable distributions (see Nolan \cite{N13});
    
    \item and the generalized Laplace kernel $\Psi(\xi)=(1+ \xi)^{-b}, \ b>0$ (see Kozubowski et al. \cite{KPR13}).
\end{itemize}
In the present work, we shall restrict our tests by making use only of the standard normal kernel. For a more in-depth discussion on kernels, we refer to Micchelli et al. \cite{MXZ06}.


\section{Test implementation by re-sampling} \label{sec3}  
When some of the component parameters occurring in $\bvtheta$ can be standardized out, the asymptotic null distribution of the proposed test statistic $\widehat T^{(\Psi)}_{n,m}$ in \eqref{ts} does not depend on them. Such parameters are typically location and scatter parameters, while others, labeled as shape parameters, such as skewness and kurtosis, cannot usually be standardized out and, therefore, will ultimately affect the asymptotic null distribution of the test statistic (see Meintanis and Swanepoel \cite{MS07}). In such cases, we can decompose the parameter vector as $\bvtheta=(\btheta, \blambda)$, where $\btheta$ denotes the non-shape parameters and $\blambda$ denotes the part of $\bvtheta$ that contains only shape parameters. In the presence of a canonical form of the distribution under test (see Section \ref{sec4}), the asymptotic null distribution of the test statistic may be simulated by setting $\btheta=\btheta_0$, where $\btheta_0$ is some standard value of $\btheta$, and $\blambda$ is set equal to its value $\widetilde{\blambda}:=\widetilde{\blambda}(\blambda,\btheta)$ in the canonical form.  

In the following, we outline the re-sampling procedure used within this work to approximate the test statistic's asymptotic distribution under the null hypothesis and indicate how the test can be carried out in practice. For definiteness, and for fixed $(n,m,\Psi)$, we write the test statistic in \eqref{psi} as $T({\cal{X}}_n; {\cal{X}}_{0,m})$, where ${\cal{X}}_n=(\bX_1, \ldots, \bX_n)$ denotes the observed data and ${\cal{X}}_{0,m}=(\bX_{0,1},\ldots,\bX_{0,m})$ denotes the data generated from ${\cal {F}}_{\widehat{\bvtheta}_n}$, i.e.,  from the null distribution with estimated parameters. We consider two cases of null hypotheses, one ``composite" with all parameters being estimated, while the other will be labeled ``simple", although in this second case too, some, but not all, parameters are estimated. 


\subsection{Simple null hypothesis}\label{subsec31}

Here we are interested in the goodness-of-fit testing problem associated with the simple null hypothesis  
\be \label{null3s} 
 {\cal{H}}^{\rm s}_{0}: \mbox{the law of} \ \bX  \in {\cal{F}}_{\btheta,\blambda}, \ \mbox{for a fixed} \ \blambda=\blambda_0, \mbox{ and for some} \ \btheta;
\ee
and alternative hypothesis
\be \label{alt3s} 
 {\cal{H}}^{\rm s}_{1}: \mbox{the law of} \ \bX  \not \in {\cal{F}}_{\btheta,\blambda}, \ \mbox{for a fixed} \ \blambda=\blambda_0, \mbox{ and any} \ \btheta. 
\ee

Although we labeled \eqref{null3s} as a simple null hypothesis, it should be pointed out that the parameter $\btheta$ is left unspecified in ${\cal{H}}^{\rm s}_{0}$, and that our test procedure incorporates an estimation step for this parameter. For this case, the computation of critical points is based on simple Monte Carlo sampling from the distribution figuring in the null hypothesis. The steps of this Monte Carlo run are as follows:
\begin{enumerate}[label=Step \arabic* -]
  \item Generate a random sample ${\cal{X}}_n = \{\bX_1, \ldots, \bX_n\}$ from ${\cal{F}}_{\btheta_0,\blambda_0}$, compute the estimate $\widehat {\btheta}_n$, and obtain the standardized sample $\widehat{\cal{X}}_n = (\widehat{\bX}_1, \ldots, \widehat{\bX}_n)$, where $\widehat {\bX_j}=\widehat {\bX_j}({\bX_j},\widehat {\btheta}_n), \ j \in \{1, \ldots, n\}$. 
  
  \item Generate a random sample $\X_{0,m} = \{\bX_{0,1}, \ldots, \bX_{0,m}\}$ from ${\cal{F}}_{\btheta_0,\blambda_0}$. 
  
  \item Compute the test statistic $T := T^{\Psi}(\widehat{\cal{X}}_n, \X_{0,m})$, according to \eqref{psi}.
  
  \item Repeat Steps 1-3 several times, say $M$, and obtain the set of test statistics $\{T_1, \ldots, T_M$\}. Then the critical point, say $\widehat c_\delta$, is defined as the $(1 - \delta)\%$ quantile of $(T_m, \ m=1,\ldots,M)$.
\end{enumerate}
  
Having obtained the empirical critical point, $\widehat c_\delta$ is used to compute the test's empirical powers. In this connection, we generate a random sample $\X_n = \{\bX_1, \ldots, \bX_n\}$ from any distribution belonging to the set of alternatives in the alternative hypothesis ${\cal{H}}^{\rm s}_{1}$, and perform Steps 1-3 above, thereby computing the test statistic $T$. We reject the null hypothesis ${\cal{H}}^{\rm s}_{0}$ if $T>\widehat c_\delta$. We repeat this procedure several times, say $L$, and obtain the empirical power rate as $L^{-1} \sum_{\ell=1}^L \mathds{1}_{T_{\ell} > \widehat c_\delta}$, where $T_\ell$ denotes the test statistic corresponding to the $\ell^{\rm th}$ sample, for $\ell \in \{1, \ldots, L\}$.


\subsection{Composite null hypothesis}\label{subsec32}

Here we are interested in the (fully) composite goodness-of-fit testing problem whereby all distributional parameters are estimated from the observed data. For reasons of explicitness, we state the null hypothesis as
\be \label{null3c} 
 {\cal{H}}^{\rm c}_0: \mbox{the law of} \ \bX  \in {\cal{F}}_{\btheta,\blambda}, \quad \mbox{for some} \ (\btheta,\blambda);
\ee
as well as the alternative 
\be \label{alt3c} 
 {\cal{H}}^{\rm c}_1: \mbox{the law of} \ \bX  \not \in {\cal{F}}_{\btheta,\blambda}, \quad \mbox{for any} \ (\btheta,\blambda).
\ee

For this case, the re-sampling scheme is as follows. On the basis of ${\cal {X}}_n$, compute the estimator $\widehat {\bvtheta}_n=(\widehat{\btheta}_n, \widehat {\blambda}_n)$ of $\bvtheta$ and standardize the observations as $\widehat{\bX}_j = \widehat {\bX}_j(\bX_j,\widehat{\btheta}_n)$, $j \in \{1, \ldots, n\}$. Now generate a random sample $\widehat {\cal{X}}_{0,m}:=(\widehat \bX_{0,1},\ldots, \widehat \bX_{0,m})$ under the null hypothesis with $(\btheta, \blambda)$ set equal to $(\btheta_0, \widehat {\widetilde {\blambda}}_n)$, where $\widehat {\widetilde {\blambda}}_n=\widehat {\widetilde {\blambda}}_n(\widehat {\btheta}_n, \widehat {\blambda}_n)$ is the parameter estimate of $\blambda$ induced by the parametrization. Then the value of the original test statistic is computed according to \eqref{psi} as $T = T(\widehat{\cal{X}}_n; \widehat {\cal{X}}_{0,m})$, where $\widehat {\cal{X}}_{n}:=(\widehat \bX_{1},\ldots, \widehat \bX_{n})$. In turn, the critical point against which the value of $T$ will be compared is computed using a {\it{parametric bootstrap}} procedure, the steps of which are outlined below: 

\begin{enumerate}[label=Step \arabic*:]
 \item Generate a random sample ${\cal{X}}^*_{n}:=(\bX^*_{1}, \ldots, \bX^*_{n})$ under the null hypothesis with $(\btheta, \blambda)$ set equal to $(\btheta_0, \widehat {\widetilde {\blambda}}_n)$. 

 \item On the basis of ${\cal {X}}^*_n$, compute the estimator $\widehat {\bvtheta}^*_n=(\widehat {\btheta}^*_n, \widehat {\blambda}^*_n)$.  

 \item Standardize the components of ${\cal {X}}^*_n$ as $\widehat{\bX}^*_j = \widehat {\bX}^*_j(\bX^*_j,\widehat {\btheta}^*_n)$, $j \in \{1, \ldots, n\}$. 

 \item Generate a random sample $\widehat {\cal{X}}^*_{0,m}:=(\widehat \bX^*_{0,1},\ldots, \widehat \bX^*_{0,m})$ under the null hypothesis with $(\btheta, \blambda)$ set equal to $(\btheta_0, \widehat {\widetilde {\blambda}}^*_n)$, where $\widehat {\widetilde {\blambda}}^*_n=\widehat {\widetilde {\blambda}}^*(\widehat {\btheta}^*_n, \widehat {\blambda}^*_n)$ is the bootstrap parameter estimate of $\widehat {\widetilde {\blambda}}_n$. 

 \item Compute the value of the bootstrap test statistic by  \eqref{psi} as $T^* = T(\widehat{\cal{X}}^*_n; \widehat {\cal{X}}^*_{0,m})$, where $\widehat {\cal{X}}^*_{n}:=(\widehat \bX^*_{1},\ldots, \widehat \bX^*_{n})$. 
 
 \item Steps 1-5 are repeated several times, say $B$, and thereby we compute the $(1-\delta)\%$ quantile $c_\delta$, with $\delta \in (0,1)$, of the empirical distribution of $(T^*_{b}, b=1,\ldots, B)$ as the size-$\delta$ critical value of the test statistic. 

 \item Repeat Steps 1-6 several times, say $M$, and thereby obtain pairs of test statistics and corresponding bootstrap critical points $(T_m,c_{\delta,m})$, $m \in \{1,\ldots, M\}$.
 
 \item Compute the empirical rejection rate as $M^{-1} \sum_{m=1}^M \mathds{1}_{T_m > c_{\delta,m}}$.
\end{enumerate}

Because the above parametric bootstrap procedure is time-consuming, we adopt the {\it warp-speed bootstrap} method of Giacomini et al. \cite{GPW13}. Thus, rather than computing a critical value $c_{\delta,m}$ for each of the $M$ Monte Carlo samples, we produce a single critical value that is used for all Monte Carlo samples. To do so, we generate only one single bootstrap sample, i.e., with $B = 1$ on Step~6, for each of the $M$ Monte Carlo samples and compute the corresponding bootstrap test statistic, say $T_m^*$, from this single bootstrap sample. Then the warp-speed critical value, say $\widetilde c_\delta$, is computed from $(T_m^*, \ m \in \{1, \ldots, M\})$ analogously as in Step 6 above, and the empirical rejection rate is given by $M^{-1} \sum_{m=1}^M \mathds{1}_{T_m> \widetilde c_{\delta}}$.  


\section{Families of skewed distributions} \label{sec4}
In this section, we consider a collection of five families of skewed distributions and exemplify how our method, described in Sections \ref{sec2} and \ref{sec3}, may be applied to perform goodness-of-fit tests with them. In the following Sections \ref{sec5} and \ref{sec6}, we shall use these five families of distributions, respectively, in simulation studies and applications to real data sets.

\subsection{Multivariate skew-normal distribution}
The multivariate skew-normal (SN) distribution may be conveniently defined by the CF (see Azzalini and Dalla Valle \cite{AD96})
\be
\varphi_\bvtheta(\bt)=2 e^{{\rm{i}}\bt^\top \bxi-\frac{1}{2} \bt^\top \bOmega \bt }\:\Phi\left({\rm{i}}\frac{\balpha^\top \bOmega \bt}{\sqrt{1+\balpha^\top \bOmega \balpha}}\right), 
\ee
where $\Phi(\cdot)$ is the standard normal cumulative distribution function, $\bvtheta:=(\bxi,\bOmega, \balpha)$ is the associated parameters vector, with $(\bxi,\balpha) \in \mathbb R^p \times \mathbb R^p$ being, respectively, the location and skewness parameters, and where $\bOmega \in \R^{p \times p}$ is a symmetric positive definite matrix. We shall write SN$_p(\bxi,\bOmega,\balpha)$ to denote this distribution, with $\balpha=\0$ rendering the $p$-variate normal distribution with mean $\bxi$ and covariance matrix equal to $\bOmega$. To apply our test, a consistent estimator of the parameters in $\bvtheta$ is required. There exist a variety of methods for estimating them, including maximum-likelihood (MLE) and moments-based estimation methods (see, e.g., Azzalini and Capitanio \cite{AC99}, Azzalini et al. \cite{AGS10}, and Flecher et al. \cite{FNA09}), as well as packages available for the same purpose (see Azzalini \cite{A22}). There also exist a few goodness-of-fit tests in this case, namely, the tests of Balakrishnan et al. \cite{BCS14}, Gonz\'alez-Estrada et al. \cite{GVA22}, Jim\'enez-Gamero and Kim \cite{JK15}, and Meintanis and Hl\'avka \cite{MH10} which will be discussed further down in the paper (see Subsection \ref{subsec54}).

It may be shown that if $\bX\sim {\rm{SN}}_p(\bxi,\bOmega,\balpha)$, then there exists a matrix $\bH \in \R^{p \times p}$ such that $\bH^\top (\bX-\bxi) \sim {\rm{SN}}_p(\0,\bI, \balpha^*)$, with $\0$ and $\bI$ denoting, respectively, the zero vector and identity matrix in the indicated dimension, and $\balpha^*= (\alpha^*, 0, \ldots, 0)^\top$, with $\alpha^* = (\balpha ^\top \bar{\bOmega} \balpha)^{1/2}$ and $\bOmega = \bm{\omega} \bar{\bOmega} \bm{\omega}$, where $\bm{\omega} = \mbox{diag}(\omega_1, \ldots, \omega_p)$ is a positive-definite scale matrix (see eq. (5.2) in Azzalini and Capitanio \cite{AC14}). In the literature (see, e.g., Azzalini and Capitanio \cite{AC14} and Capitanio \cite{C20}), the ${\rm{SN}}_p(\0,\bI, \balpha^*)$ distribution is also called the canonical form. Moreover, it was proved by Capitanio \cite{C20} that the choice of 
\be \label{matrixH}
\bH = \bOmega^{-1/2} \bQ, 
\ee
where $\bOmega^{-1/2}$ is the unique inverse matrix of the positive definite symmetric square root matrix of $\bOmega$, and $\bQ$ is obtained through the spectral decomposition $\bQ \bm{\Lambda} \bQ^\top = \bOmega^{-1/2} \bSig \bOmega^{-1/2}$, with $\bSig$ being the covariance matrix of $\bX$, leads to the conclusion that $\bH^\top (\bX-\bxi)$ follows a canonical skew-normal distribution. In this connection, write $\widehat \bvtheta_n = (\widehat {\bxi}_n,\widehat{\bOmega}_n, \widehat{\balpha}_n)$ for an estimator of $\bvtheta$, and consider the standardized observations $\widehat{\bX}_j = \widehat {\bH}_n^\top (\bX_j-\widehat {\bxi}_n), j \in \{1,\ldots,n\}$, where $\widehat{\bH}_n = \widehat{\bOmega}_n^{-1/2} \widehat{\bQ}_n$. Then the test figuring in \eqref{psi} is readily applied by replacing $\bX_j$ by $\widehat \bX_j$, and where the $\bX_{0,j}$ are drawn from a  SN distribution with parameters $(\bxi, \bOmega)=(\0,\bI)$ and $\balpha^*$, the latter being replaced by $\widehat{\balpha}_n^* = (\widehat{\alpha}_n^*, 0, \ldots, 0)^\top$, where $\widehat{\alpha}_n^* = (\widehat{\balpha}_n^\top \widehat{\bar{\bOmega}}_n \widehat {\balpha}_n)^{1/2}$. For obtaining the estimates of $\widehat{\bxi}_n$, $\widehat{\bH}_n$, and $\widehat{\alpha}_n^*$, we suggest the use of the \texttt{sn} \cite{A22} package within the \textsf{R} \cite{R22} software environment.

\subsection{Multivariate skew-t distribution}

The multivariate skew-t (ST) distribution is related to the multivariate skew-normal distribution through the stochastic equation $\bY = \bxi + \sqrt{\eta} \bX$, where $\bX$ has a multivariate skew-normal distribution, $\bX \sim {\rm{SN}}_p(\0,\bOmega,\balpha)$, and $\eta$ has an inverse-Gamma distribution with shape and scale parameters both equal to $\nu/2$, i.e., $\eta \sim {\rm IG}(\nu/2, \nu/2)$. It was shown by Kim and Genton \cite{KG11}, theorem 7, that the CF of $\bY$ is given by
\be
\varphi_\bvtheta(\bt) = \exp({\rm i} {\bt}^\top \bxi) [\psi_{T_p}(\bOmega^{1/2} \bt) + {\rm i} \tau^+ (\brho, \omega \bt)],
\ee
where 
\begin{align*}
   \brho & = \bOmega \balpha/ (1 + \balpha^\top \bOmega \balpha)^{1/2},\\
   \psi_{T_p}(\bt) &= \frac{\|\sqrt{\nu} \bt\|^{\nu/2}}{\Gamma(\nu/2) \, 2^{\nu/2 - 1}} \ K_{\nu/2} (\|\sqrt{\nu} \bt\|), \quad \mbox{for } \bt \in \R^p, \ \nu > 0,\\
   \tau^+(\brho, \omega \bt) &= \int _0^\infty \exp(-x \bt^\top \bOmega \bt /2) \, \tau (\sqrt{x} \brho^\top \omega \bt)\, \mbox{d} H(x), \quad \mbox{for } \brho^\top \omega \bt > 0,
\end{align*}
with $\tau^+(\brho, -\omega \bt) = - \tau^+(\brho, \omega \bt)$, $\tau(x) = \int_0^x \sqrt{2/\pi} \exp(u^2/2)\, {{\rm d}}u$, for $x >0$, with $\tau(-x) = -\tau(x)$, $H(x) = \Gamma(\nu/2, \nu/(2x))/\Gamma(\nu/2)$, for $x > 0$, denoting the cumulative distribution function of $\eta$, with $\Gamma(a) = \Gamma(a, 0)$ and where $\Gamma(a, b) = \int_b^\infty t^{a-1} e^{-t} \mbox{d}t$, for $b \geq 0$, represents the upper incomplete Gamma function, and $K_{\lambda}(\cdot)$ is the integral representation of the modified Bessel function of the third kind, defined as $K_\lambda(w) = \frac{1}{2} \int_0^\infty x^{\lambda - 1} \exp\left\{- \frac{w}{2} \left(x + \frac{1}{x}\right)\right\} \mbox{d} x$, for $w > 0$ and $\lambda \in \R$. Here $\bvtheta = (\bxi, \bOmega, \balpha, \nu)$ and we write $\bY \sim {\rm ST}_p(\bxi, \bOmega, \balpha, \nu).$

Since the multivariate skew-t distribution can be expressed as a scale mixture of a skew-normal distribution, proposition 2 in Capitanio \cite{C20} guarantees that by taking once again the matrix $\bH$ as defined in \eqref{matrixH}, any random vector $\bY \sim$ ST$_p(\bxi, \bOmega, \balpha, \nu)$ can be transformed into the canonical skew-t distribution $\bH^\top \bY \sim$ ST$_p(\0, \bI, \balpha_\bY^*, \nu)$. In particular, $\bxi$ and $\bOmega$, like in the skew-normal case, are nuisance parameters so that they can be dismissed for hypothesis testing after the standardization is performed. Additionally, for the simulation studies in Section \ref{sec5}, it will suffice to implement the tests for different choices of $p$, $\nu$, and $\alpha^*$ (the unique non-null component of $\balpha_\bY^* = (\alpha^*, 0, \ldots, 0)^\top$), substantially reducing the cases that need a proper investigation. Here the \texttt{sn} \cite{A22} package in \textsf{R} \cite{R22} can also be used to retrieve the desired multivariate skew-t's parameters estimates as it was also designed for this purpose. The hypothesis test is then carried out similarly to the skew-normal case discussed in the preceding subsection.

\subsection{Multivariate skew-Laplace distribution}

The multivariate skew-Laplace (SL) distribution may be conveniently defined by the CF (see Arslan \cite{A10})
\be \label{msl}
\varphi_\bvtheta(\bt) = \frac{e^{{{\rm i}}\bt^\top \bxi}}{(1 + \bt^\top \bOmega \, \bt - 2 {{\rm i}} \bt^\top \balpha)^{(p+1)/2}}, 
\ee
with $\bvtheta = (\bxi, \bOmega, \balpha)$, where $(\bxi,\balpha) \in \mathbb R^p \times \mathbb R^p$ are, respectively, location and skewness parameters, and $\bOmega$ is a symmetric positive definite matrix. We will use the notation SL$_p(\bxi,\bOmega,\balpha)$ for this distribution. Although the multivariate skew-Laplace distribution proposed by Arslan \cite{A10} has very similar properties to the distinct version introduced by Kotz et al. \cite{KKP01}, Arslan's alternative has a simpler probability density function, allowing for uncomplicated estimation methods of its parameters in the multivariate setting. Also, Arslan \cite{A10} proposed an efficient EM algorithm that can be used for the estimation of $\bxi, \bOmega,$ and $\balpha$.

Analogously to the skew-normal and skew-t distributions, provided that $\balpha \neq \0$, it may be shown that $\bX\sim {{\rm SL}}_p(\bxi,\bOmega,\balpha)$ can be reduced to a canonical form. This novel result brings down the burden of testing for nuisance parameters, as well as reducing the skewness to a singular one-dimensional component.

\begin{prop}[Canonical form - SL$_p$ distribution]\label{propcanLap}
Let $\bX \sim \mbox{SL}_p(\bxi, \bOmega, \balpha)$ with $\balpha \neq \0$ and consider the affine non-singular transform 
\begin{equation}
  \bX^* = \bH^\top (\bX - \bxi),
\end{equation}
with $\bH = \bOmega^{-1/2} \bQ$, where $\bOmega^{-1/2}$ denotes the inverse of the unique positive definite symmetric square root matrix of $\bOmega$, and $\bQ = [\bv_1 \ldots \bv_p]$ is the orthogonal matrix with $\bv_1 = \bOmega^{-1/2} \balpha/ \|\bOmega^{-1/2} \balpha\|$ as its first column vector and the remaining columns $\bv_2, \ldots, \bv_p$ belong to the orthogonal complement of $\bv_1$. Then $\bX^* \sim \mbox{SL}_p(\0, \bI, \balpha_\bX^*)$ with $\balpha_{\bX}^* = (\alpha^*, 0, \ldots, 0)^\top$ and $\alpha^* = {\|\bOmega^{-1/2} \balpha\|}$.
\end{prop}
\begin{proof}
  From proposition 3 in Arslan \cite{A10}, if $\bA \in \R^{p \times p}$ is any full rank matrix, it follows that $\bA (\bX - \bxi) \sim \mbox{SL}_p(\0, \bA \bOmega \bA^\top, \bA \balpha)$. Since $\balpha = (\alpha_1, \ldots, \alpha_p)^\top \neq \0$ and $\bOmega^{-1/2}$ is non-singular, we have $\bOmega^{-1/2} \balpha \neq \0$. Hence, there is at least one component $u_k$ in $\bu = (u_1, \ldots, u_p)^\top = \bOmega^{-1/2} \balpha$ that is non-null. Moreover, if $\bm{e}_j$ represents the $j$th canonical vector in $\R^p$, for $j \in \{1, \ldots, p\}$, then $\{\bu\} \cup \{\bm{e}_j \in \R^p : j \neq k\}$ is a basis of $\R^p$. For ease of reading, let us rename the vectors $\bu_1 = \bu$, $\bu_{j+1} = \bm{e}_j$, for $1 \leq j < k$, and $\bu_{j} = \bm{e}_j$, for $k < j \leq p$. Then we can apply the Gram-Schmidt process to find an orthonormal basis $\{\bv_1, \ldots, \bv_p\}$ of $\R^p$. For this, take $\tilde{\bv}_1 = \bu_1 = \bOmega^{-1/2} \balpha$ and $\tilde{\bv}_j = \bu_j - \sum_{i=1}^{j-1} \frac{\tilde{\bv}_i^\top \bu_j }{\tilde{\bv}_i^\top \tilde{\bv}_i} \tilde{\bv}_i$, for $2 \leq j \leq p$. The desired basis is obtained from the normalizations $\bv_j = \tilde{\bv}_j / \|\tilde{\bv}_j\|$, for $j \in \{1, \ldots, p\}$. Therefore, taking $\bQ = [\bv_1 \ldots \bv_p]$ and $\bH = \bOmega^{-1/2} \bQ$, it follows that $\bH^\top (\bX - \bxi) \sim \mbox{SL}_p(\0, \bH^\top \bOmega \bH, \bH^\top \balpha)$, with $\bH^\top \bOmega \bH = \bQ^\top \bOmega^{-1/2} \bOmega\, \bOmega^{-1/2} \bQ = \bI$ and $\bH^\top \balpha = \bQ^\top \bOmega^{-1/2} \balpha = (\|\bOmega^{-1/2} \balpha\|, 0, \ldots, 0)^\top$.
\end{proof}

Analogously to the SN and ST cases, the fact that a canonical form is available for the skew-Laplace distribution is useful in the implementation of the goodness-of-fit test due to a reduced number of parameters to be considered. Moreover, the proof of Proposition \ref{propcanLap} gives at the same time an algorithm to find the canonical form of a multivariate SL$_p$ distribution. Taking this in consideration, let $\widehat{\bX}_j = \widehat \bH_n^\top (\bX_j-\widehat {\bxi}_n)$, $j \in \{1,\ldots,n\}$, be the standardized observations. Then the test figuring in \eqref{psi} is applied by replacing $\bX_j$ by $\widehat{\bX}_j$, where $\bX_{0,j}$ are drawn from a  SL distribution with parameters $(\bxi, \bOmega)=(\0,\bI)$ and $\balpha$ is set equal to $(\|\widehat {\bOmega}_n^{-1/2} \widehat{\balpha}_n\|, 0, \ldots, 0)^\top$. To obtain the estimates of $\widehat{\bxi}_n$, $\widehat{\bOmega}_n$, and $\widehat{\balpha}_n$, we suggest the EM algorithm proposed by Arslan \cite{A10}.

\subsection{Multivariate Tukey \texorpdfstring{$\bg$-and-$\bh$}{g-and-h} distribution}

The multivariate Tukey $g$-and-$h$ distribution (GH) was first introduced by Field and Genton \cite{FG06} as a generalization of its univariate counterpart presented by Tukey in 1977. It has been gaining popularity due to its flexible marginal distributions, allowing for the fitting of skewed and heavy-tailed data sets from climate and environmental problems (see, e.g., Jeong et al. \cite{JYCG19}, Yan and Genton \cite{YG19}, and Yan et al. \cite{YJG20}). Given two parameter vectors $\bg = (g_1, \ldots, g_p)^\top \in \R^p$ and $\bh = (h_1, \ldots, h_p)^\top \in \R^p_+$, the random vector $\bY \in \R^p$ is said to have a standard multivariate Tukey $g$-and-$h$ distribution if it can be represented as 
\begin{equation}
  \bY = {\btau}_{\bg, \bh}(\bZ):= (\tau_{g_1, h_1}(Z_1), \ldots, \tau_{g_p, h_p}(Z_p))^\top,
\end{equation}
where $\bZ = (Z_1, \ldots, Z_p)^\top \sim {N}_p(\0, \bI)$ has a standard multivariate normal distribution and, for two given $g \in \R$ and $h \in \R_+$, the univariate function $\tau_{g, h}$ is defined as
\begin{equation}
  \tau_{g, h}(z) =
  \begin{cases}
    \displaystyle \left(\frac{\exp(g z) - 1}{g} \right) \exp \left(\frac{h z^2}{2}\right), & g \in \R \setminus \{0\},\\[3mm]
    \displaystyle z \exp \left(\frac{h z^2}{2}\right), & g = 0,
  \end{cases}
\end{equation}
for any $z \in \R$. The general multivariate Tukey $g$-and-$h$ distribution is then defined as \begin{equation}
  \bY = \bOmega \, \btau_{\bg, \bh}(\bZ) + \bxi,
\end{equation} 
where $\bOmega \in \R^{p \times p}$ and $\bxi \in \R^p$ are, respectively, a positive definite matrix and a location vector. Here $\bvtheta = (\bxi, \bOmega, \bg, \bh)$ and the nuisance parameters are $\bxi$ and $\bOmega$ as the can be standardised out. We shall use the notation GH$_p(\bxi, \bOmega, \bg, \bh)$. 

In contrast with the SN, ST, and SL distributions, the Tukey $g$-and-$h$ distribution does not have a known transformation that allows one to represent it in a reducible canonical form. Also, since the inverse of $\tau_{g,h}(\cdot)$ does not have a closed expression, classical estimation methods, such as the MLE method, rely on numerical approximations. Likewise, the CF, probability density function, and cumulative distribution function can only be computed numerically. Several different techniques have been proposed for the estimation and fitting of the univariate Tukey $g$-and-$h$ distribution (see the review paper by M\"{o}stel et al. \cite{MFPP21}). However, for the multivariate case, only a few methods are available. For instance, Field and Genton \cite{FG06} used multivariate quantiles for data fitting. He and Raghunathan \cite{HR12} assumed that $\bOmega$ is a diagonal matrix and, for this reason, they proposed an algorithm that uses quantiles from the univariate Tukey $g$-and-$h$.

In this work, we opted for the MLE method. For this purpose, we need to find the parameters $\bvtheta$ that maximize the log-likelihood function
\begin{equation}
  \ell(\bvtheta | \bY_1, \ldots, \bY_n) = - n \ln |\bOmega| + \sum_{i=1}^n \sum_{j=1}^p \left[\ln \left\{ \phi(u_i^{(j)})\right\} - \ln \left\{ \tau'_{g_j,h_j}(u_i^{(j)})\right\}\right],
\end{equation}
where $u_i^{(j)} = \tau^{-1}_{g_j, h_j} \left(\{\bOmega^{-1} (\bY_i - \bxi)\}^\top \bm{e}_j \right)$, for $i \in \{1, \ldots, n\}$ and $j \in \{1, \ldots, p\}$, $\bm{e}_j$ denotes the $j$th canonical vector in $\R^p$, $\phi(\cdot)$ is the pdf of a standard univariate normal distribution, and $\tau'_{g_j,h_j}(z)$ denotes the derivative of $\tau_{g_j,h_j}(z)$ concerning its argument $z$. To give an approximated value of the quantities $u_i^{(j)}$, we use the \texttt{uniroot} function available in the \texttt{stats} \cite{R22a} package in \textsf{R} \cite{R22}. For the maximization procedure, we use the \texttt{optim} function available in the \texttt{MASS} \cite{R22b} package in \textsf{R} \cite{R22}, together with the ``Nelder-Mead'' method.

\subsection{Multivariate \texorpdfstring{$\balpha$}{alpha}-stable distributions}

The $\alpha$-stable distributions (AS), similar to the multivariate skew-t and Tukey $g$-and-$h$ distributions, are another possible extension of the multivariate Gaussian distribution that comports skewness. However, while the multivariate skew-t family has finite second moments, the $\alpha$-stable distributions are regulated by a parameter $\alpha \in (0, 2]$, called the tail index, and it has only finite second moments if $\alpha = 2$, which reduces itself to the multivariate Gaussian case. Particularly, if $\bX$ has a multivariate $\alpha$-stable distribution, then $\E(\|\bX\|^s) < \infty$, if $0 < s < \alpha$, and $\E(\|\bX\|^s) = \infty$, if $s \geq \alpha$. Several parameterizations for the CF have been proposed in the literature. Here, for numerical reasons, we adopt the S$^0$ parameterization introduced in Abdul-Hamid and Nolan \cite{AN98}, for which the CF is given by
\begin{equation}\label{ch.f.}
 \varphi_\bvtheta(\bt) = \exp\left(-\int_{\Sp} \psi_\alpha (\bt^\top \bs)\, \Gamma(d\bs) +{\rm i} \bt^\top \bxi \right), \quad \bt \in \R^p,
 \end{equation}
 where
 \begin{equation}\label{psi0}
 \psi_\alpha(u) =
 \begin{cases}
  |u|^\alpha\Big(1+{\rm i}\,\mbox{sign}(u)\tan\left(\frac{\pi\alpha}{2}\right)\left(|u|^{1-\alpha}-1\right)\Big), & \alpha\neq 1,\\
  |u|~\Big(1+{\rm i}\frac{2}{\pi}\mbox{sign}(u) \ln|u|\Big), & \alpha=1.
  \end{cases}
 \end{equation}
In this case, we denote $\bX \sim \rS_{\alpha}^0(\Gamma, \bxi)$ to indicate that the random vector $\bX$ has an $\alpha$-stable distribution with finite spectral measure $\Gamma(\cdot)$, defined on the unitary sphere $\Sp := \{\bs \in \R^p : \|\bs\| = 1\}$, and shift vector $\bxi \in \R^p$. We note from \eqref{ch.f.} that the multivariate $\alpha$-stable distributions are a semi-parametric family, being completely defined by the triplet $(\alpha, \Gamma, \bxi)$ and belonging to the more general class of infinitely divisible distributions. Furthermore, $\bX - \bxi \sim \rS_\alpha^0(\Gamma, \0)$, so that $\bxi$ is a nuisance parameter vector and it coincides with the mean vector when $1 < \alpha \leq 2$. 

Byczkowski et al. \cite{BNR93} proved that an approximation can be given for the spectral measure $\Gamma$, being useful for numerical computations and simulations. This approximation is described as follows: consider a finite partition $P = \{A_1, \ldots, A_m\}$ of $\Sp$ and a set of points $S = \{\bs_1, \ldots, \bs_m\} \subseteq \Sp$. Then, by setting $\gamma_i = \Gamma(A_i)$, the discrete spectral measure associated to $(P, S, \Gamma)$ is defined as
 \begin{equation} \label{dmeasure}
 \Gamma^*(\cdot) = \sum_{i=1}^m \Gamma(A_i) \, \mathds{1}_{\bs_i}(\cdot) = \sum_{i=1}^m \gamma_i \, \mathds{1}_{\bs_i}(\cdot),
 \end{equation}
for any Borelian in $\Sp$. Here $\Gamma^*$ is implicitly defined by concentrating the mass of $\Gamma(A_i)$ at each point $\bs_i$, the only requirement being that $\sup_{\bs \in A_i} |\bs - \bs_i|$ is sufficiently small, for each $i = 1, \ldots, m$. Hence, the discrete spectral measure $\Gamma^*$ can be used in practice instead of its continuous counterpart $\Gamma$. For simulating multivariate $\alpha$-stable random vectors, we shall use the result from Modarres and Nolan \cite{MN94} which states that, if $\bX \sim \rS_\alpha^0(\Gamma^*, \bxi)$, with $\Gamma^*(\cdot)$ defined in \eqref{dmeasure}, then
\begin{equation}\label{gen}
 \bX =
 \begin{cases}
 \sum_{i=1}^m \gamma_i^{1/\alpha} Z_i \bs_i + \tilde{\bxi}, & \alpha \neq 1,\\
 \sum_{i=1}^m \gamma_i(Z_i+ \frac{2}{\pi} \ln(\gamma_i)) \bs_i + \tilde{\bxi},& \alpha=1,
 \end{cases}
 \end{equation}
where $\tilde{\bxi} = \bxi - \tan\left(\frac{\pi \alpha}{2}\right) \sum_{i=1}^m \gamma_i \bs_i$ and $Z_1, \ldots, Z_m$ are i.i.d.\ one-dimensional $\alpha$-stable random variables with $Z_i \sim \rS_\alpha(1,1,0)$ (i.e., scale = skewness = 1 and location = 0). As for the estimation part, we use the projection method proposed by Nolan et al. \cite{NPM01}, which relies on the projections of the multivariate samples into a specifically chosen grid of values from the unitary sphere.

As it is well known, $\bX \sim \rS_\alpha^0(\Gamma, \bxi)$ is symmetric if and only if $\Gamma(\cdot)$ is symmetric on $\Sp$. So far, no tests available in the literature have been designed specifically for the general asymmetric multivariate $\alpha$-stable distributions, except for the one presented by Meintanis et al. (2015) which covers only the symmetric case. Here we propose a test that can be used for both symmetric and asymmetric cases. Although it is not usual in the literature, we use the notation $\bX \sim {\rm AS}_p(\bxi, \Gamma, \alpha)$ to indicate that $\bX$ has an asymmetric $\alpha$-stable distribution. For a more substantial review of these and further technical details concerning multivariate $\alpha$-stable distributions, we suggest reading Karling et al. \cite{KLS23} and Samorodnitsky and Taqqu \cite{ST00}.


\section{Simulation studies} \label{sec5}

In this section, we present the results of simulation studies that were produced using the tests described in Sections \ref{sec2} and \ref{sec3} with the five families of skewed distributions introduced in Section \ref{sec4}. For these simulations, we used $M = L = 1000$ as a standard value in the steps described in Subsections \ref{subsec31} and \ref{subsec32}. Firstly, we start by calculating the empirically estimated sizes of the tests for each family under a $\delta=0.05$ designed nominal level. Then we calculate the power of the test in two distinct situations, the simple hypothesis case, and the composite hypothesis case, respectively, within a second and third round of simulations. In the latter, we test the five families of distribution against the family of sinh-arcsinh distributions (see Jones and Pewsey \cite{JP09}). It is worth pointing out that, as the sample size $n$ increases, naturally, the tests require more computational time to run. Also, the efficiency of the test is prone to the number of parameters present in each family and the method used for their estimation. Finally, we close this section with a comparison between our test and a few competitors for the skew-normal family that was already available in the literature.

\subsection{Estimated sizes} \label{subsec51}

The values presented in Table \ref{tab1} were generated with the test described in Section \ref{sec3} corresponding to the composite null hypothesis case, for $\delta=0.05$ and $M=1000$ replications. The dimension considered for the samples in the tests is $p = 2$. An analogous table for $p=3$ can be found in the Supplement. For simplifying the simulations, in each case, we fixed the value of $m$ to be equal to $n$, the sample size, with $n \in \{100, 250, 500, 750, 1000\}$. In a later section, and to have a better understanding of the effect of the size $m$ of the artificial sample, we present simulation results where we fix the sample size $n$ and let $m$ vary; see Simulation \ref{sim12} in Subsection \ref{subsec54}. 

\begin{simul}\label{sim1}
  We simulated $n$ observations from an SN$_p(\0, \bI, \balpha)$ distribution, with $\balpha = (3, 0, \ldots, 0)^\top$. We calculated the empirical sizes of the test and the results are presented in Table \ref{tab1}. We notice from this table that, as the sample size increases, the empirical sizes of the test stabilize around $0.05$, which corresponds to the designed nominal level.
\end{simul}

\begin{simul}\label{sim2}
  Next, we simulated $n$ observations from an ST$_p(\0, \bI, \balpha, \nu)$ distribution, with $\balpha = (3, 0, \ldots, 0)^\top$ and $\nu = 5$. The estimated sizes of the test are presented in Table \ref{tab1}. Comparing this case with the one in Simulation \ref{sim1}, we note that, for small values of $n$ $(100, 250, 500)$, the estimated sizes are not that close to $0.05$ as the ones observed in the SN case, but they start to converge to the designed nominal level as we increase the value of $n$, showing consistency.
\end{simul}

\begin{simul}\label{sim3}
  Here, we simulated $n$ observations from an SL$_p(\0, \bI, \balpha)$ distribution, with $\balpha = (3, 0, \ldots, 0)^\top$. The estimated sizes of the test are presented in Table \ref{tab1}. Here similar results to the two preceding simulations can be observed; as the sample size increases, the estimated sizes of the test converge to the designed nominal level.
\end{simul}

\begin{simul}\label{sim4}
  Next, we simulated $n$ observations from a GH$_p(\0, \bI, \bg, \bh)$ distribution, with $\bg = (1, \ldots, 1)^\top$ and $\bh = (0.5, \ldots, 0.5)^\top$. The estimated sizes of the test are presented in Table \ref{tab1}. Since all marginal components of $\bh$ are equal to $0.5$, we are in a situation when the variance is not finite and heavier tails than the SN, ST, and SL cases are observed. This, in particular, is reflected in the estimated sizes of the test. We observe that for $n = 1000$, the rejection rate is equal to $0.071$, which is relatively high. This might be explained due to the wide range dispersion of the observed data sets over the tails. Another case, similar to this one, is shown in the next simulation.
\end{simul}

\begin{simul}\label{sim5}
  Finally, we simulated $n$ observations from an AS$_2(\0, \Gamma, \alpha)$ distribution, with discrete spectral measure $\Gamma(\cdot) = (1/3) \sum_{k = 1}^3 \mathds{1}_{\bs_k}(\cdot)$, where $\bs_k = (\cos(2\pi k/3), \sin(2\pi k/3))^\top$, for $k \in \{1, 2, 3\}$, and stability index $\alpha = 1.5$. For the estimation procedure, we used a grid size of $N = 24$ projections (see Nolan et al. \cite{NPM01}). The estimated sizes of the test are presented in Table \ref{tab1}. Here, like in Simulation~\ref{sim4}, the observations originate from a distribution with infinite variance and heavy tails. However, for this case,  as we observed an increasing estimated size of the test when we raised the sample size $n$ to $1000$, we generated two extra rounds of simulations with $n = 2500$ and  $n = 5000$ to ensure that it was not diverging from the designed nominal level. The rejection rates obtained for these cases were, respectively, $0.066$ and $0.070$.
\end{simul}

\begin{table}[h] 
    \centering
    \caption{Estimated sizes of the tests correspondent to Simulations \ref{sim1} - \ref{sim5}. \label{tab1}}
    \begin{tabular}{cccccc}
    \hline
         & $n=100$ & $n=250$ & $n=500$ & $n=750$ & $n=1000$\\[0.5mm]
    \hline   
        SN& 0.045 & 0.047 & 0.049 & 0.045 & 0.053\\[2mm]
        ST & 0.063 & 0.054 & 0.062 & 0.052 & 0.057\\[2mm]
        SL & 0.063 & 0.045 & 0.048 & 0.055 & 0.046\\[2mm]
        GH & 0.063 & 0.046 & 0.057 & 0.063 & 0.071\\[2mm]
        AS & 0.051 & 0.064 & 0.063 & 0.076 & 0.080\\
    \hline
    \end{tabular}
\end{table}

\subsection{Estimated power functions for the simple null hypothesis case} \label{subsec52}

Consider the goodness-of-fit problem with simple null hypotheses ${\cal{H}}^{\rm s}_0$ as given in \eqref{null3s} and alternative hypotheses ${\cal{H}}^{\rm s}_1$ as in \eqref{alt3s}. 
In the next five simulation runs (\ref{sim6}-\ref{sim9}), we calculate the empirical power functions of the test for a few cases of the GH, SL, SN, ST, and AS distributions. For the four first simulations, we considered the sample sizes of $n \in \{100, 250, 500, 750, 1000\}$ and $p \in \{2, 3\}$ for the dimension of the generated observations. The results are summarized and illustrated in Figure \ref{Fig1}.

\begin{simul}\label{sim6}
 We generated $n$ observations from a GH$_p(\0, \bI, \bg, \bh)$ distribution. For the null hypothesis, we take $\blambda_0 = (\bg_0, \bh_0)$, with $\bg_0 = (2, \ldots, 2)^\top$ and $\bh_0 = (1, \ldots, 1)^\top$. The power functions were calculated for $\blambda = (\bg, \bh)$ with $\bh = (h, \ldots, h)^\top$, for $h \in \{0.2, 0.4, 0.6, 0.8\}$, and $\bg = 2 \bh$, so that they only depend on the choice of $h$. The results are plotted in Figure \ref{Fig1}, items (a) and (b), respectively, for $p = 2$ and $p = 3$. A quick overview of the plotted functions suggests the obvious, as the sample size increases, the power also increases.  Moreover, as $h$ approaches $1$, the power functions converge to the size of the test as theoretically expected.
\end{simul}

\begin{simul}\label{sim7}
 Next, we generated $n$ observations from an SL$_p(\0, \bI, \balpha)$ distribution. For the null hypothesis, we take $\blambda_0 = \balpha_0 = (\alpha_0^*, 0, \ldots, 0)^\top$ with $\alpha_0^* = 3$. Then we calculated the power functions for $\blambda = \balpha = (\alpha^*, 0, \ldots, 0)^\top$ with $\alpha^* \in \{0, 1, 2, 5, 8, 13\}$ and the results are plotted in Figure \ref{Fig1}, items (c) and (d), respectively, for $p = 2$ and $p = 3$. We notice from the steepness present in the graphs that, for both dimensional cases, the power function is very sensitive to slight changes in $\alpha^*$, producing more power as its argument increases or decreases. Additionally, one can notice a small asymmetry on its graphs about $\alpha^* = 3$, precisely where the size of the test is located.
\end{simul}

\begin{simul}\label{sim8}
 Here, we generated $n$ observations from an SN$_p(\0, \bI, \balpha)$ distribution. For the null hypothesis, we considered $\blambda_0 = \balpha_0 = (\alpha_0^*, 0, \ldots, 0)^\top$ with $\alpha_0^*  = 3$, and $\blambda = \balpha = (\alpha^*, 0, \ldots, 0)^\top$ with $\alpha^* \in \{0, 1, 2, 5, 8, 13\}$ for calculating the empirical power functions. The resulting functions are plotted in Figure \ref{Fig1}, items (e) and (f), respectively, for $p = 2$ and $p = 3$. As one can notice from the graphs, the empirically estimated size of the test acts as an inflection point and the power functions are asymmetrically higher when $\alpha^* < 3$. This behavior might be due to the asymmetry of the distributions. The closer $\alpha^*$ is to $0$, the closest the distribution becomes to the multivariate normal distribution. The test has not much power when $\alpha^*$ increases, showing less steepness in that direction.
\end{simul}

\begin{figure}[ht!]
    \centering
    \includegraphics[scale=0.57]{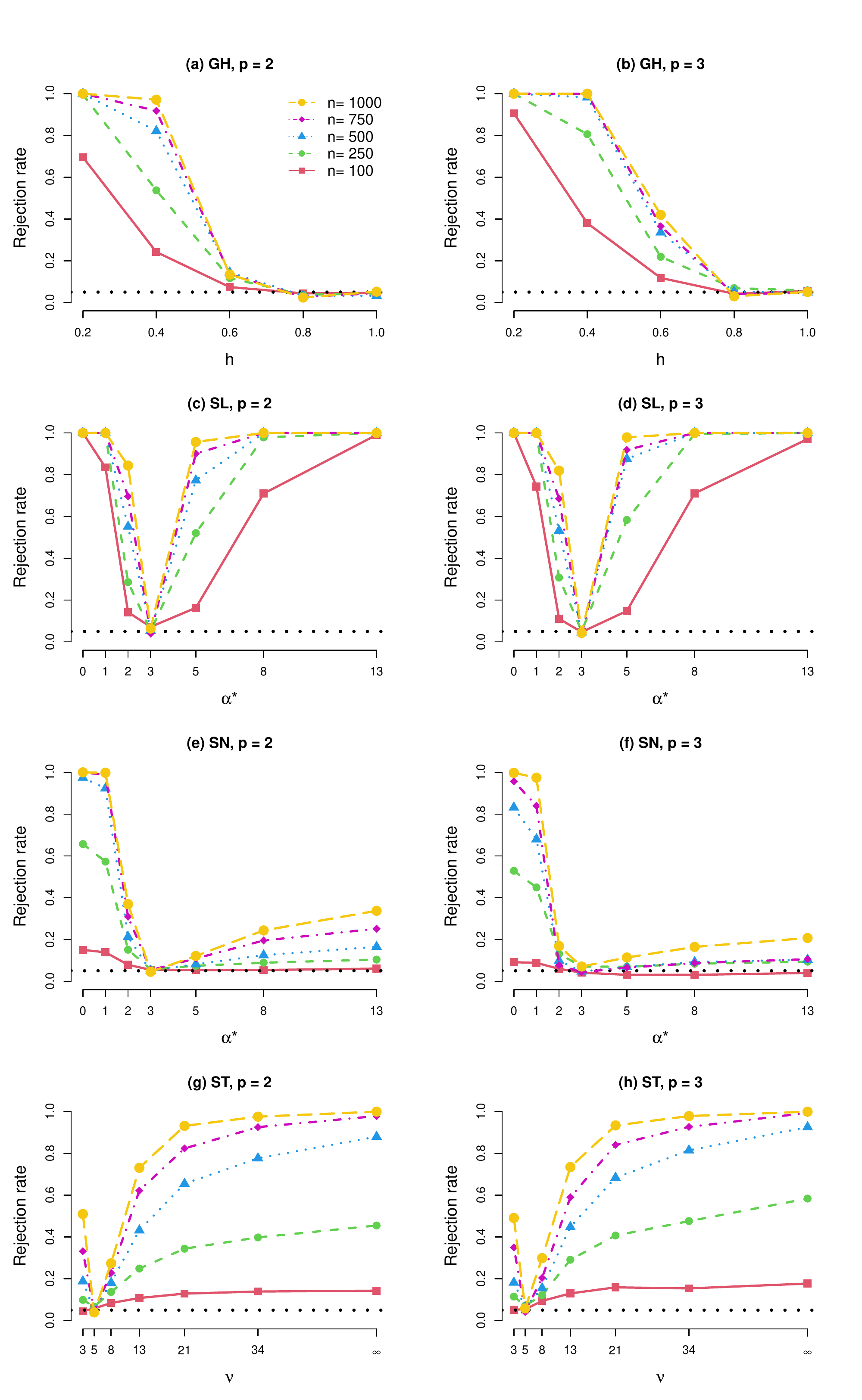}
    \caption{Power functions generated throughout Simulations \ref{sim6}-\ref{sim9}: (a),(b) Simulation \ref{sim6} $\to$ GH$_p(\0, \bI, \bg, \bh)$; (c),(d) Simulation \ref{sim7} $\to$ SL$_p(\0, \bI, \balpha)$; (e),(f) Simulation \ref{sim8} $\to$ SN$_p(\0, \bI, \balpha)$; (g);(h) Simulation \ref{sim9} $\to$ ST$_p(\0, \bI, \balpha, \nu)$. The horizontal dotted line corresponds to the $5 \%$ significance level.} \label{Fig1}
\end{figure}

\begin{simul}\label{sim9}
 Finally, we generated $n$ observations from an ST$_p(\0, \bI, \balpha, \nu)$ distribution with $\balpha = (3, 0, \ldots, 0)^\top$ fixed. We considered $\nu_0 = 5$ for the null hypothesis and $\nu \in \{3, 8, 13, 21, 34, \infty\}$ are used for calculating the empirical power functions, plotted in Figure \ref{Fig1}, items (g) and (h), respectively, for $p = 2$ and $p = 3$. By visualizing these figures, one can notice that as $\nu$ increases the power functions rapidly increase to $1$ for large sample sizes and the size of the test is attained at $\nu = 5$. Here $\nu = \infty$ is interpreted as the asymptotic distribution when $\nu \to \infty$.
\end{simul}

For the next simulation, consider $n \in \{250, 500, 750, 1000, 2500\}$ and $p = 2$.

\begin{simul}\label{sim10}
 We generated $n$ observations from an AS$_2(\0, \Gamma, \alpha)$ distribution. For the null hypothesis, we take the spectral measure $\Gamma_0(\cdot) = (1/3) \sum_{k = 1}^3 \mathds{1}_{\bs_k}(\cdot)$, where $\bs_k = (\cos(2\pi k/3), \sin(2\pi k/3))^\top$, for $k \in \{1, 2, 3\}$, and $\alpha_0 = 1.5$. Then we take two distinct sets for the alternative hypotheses. In the first set, we fix the stability index $\alpha = \alpha_0$ and calculate the power functions for the alternative spectral measures $\Gamma_q(\cdot) = (1/q) \sum_{k = 1}^q \mathds{1}_{\bs_k}(\cdot)$, with $\bs_k = (\cos(2\pi k/q), \sin(2\pi k/q))^\top$, for $k \in \{1, \ldots, q\}$ and $q \in \{34, 21, 13, 8, 5\}$. The resulting estimated power functions are plotted in Figure \ref{Fig2} (a). In the second set, we fix the spectral measure as $\Gamma(\cdot) = (1/3) \sum_{k = 1}^3 \mathds{1}_{\bs_k}(\cdot)$ and vary the parameter $\alpha$ instead, for $\alpha \in \{1.1, 1.3, 1.7, 1.9\}$. The resulting estimated power functions for this case are plotted in Figure~\ref{Fig2}~(b). In both figures, we also plotted the estimated size of the test with $\alpha = \alpha_0$ and $\Gamma(\cdot) \equiv \Gamma_3(\cdot).$
\end{simul}

\begin{figure}[h!]
    \centering
    \includegraphics[scale=0.4]{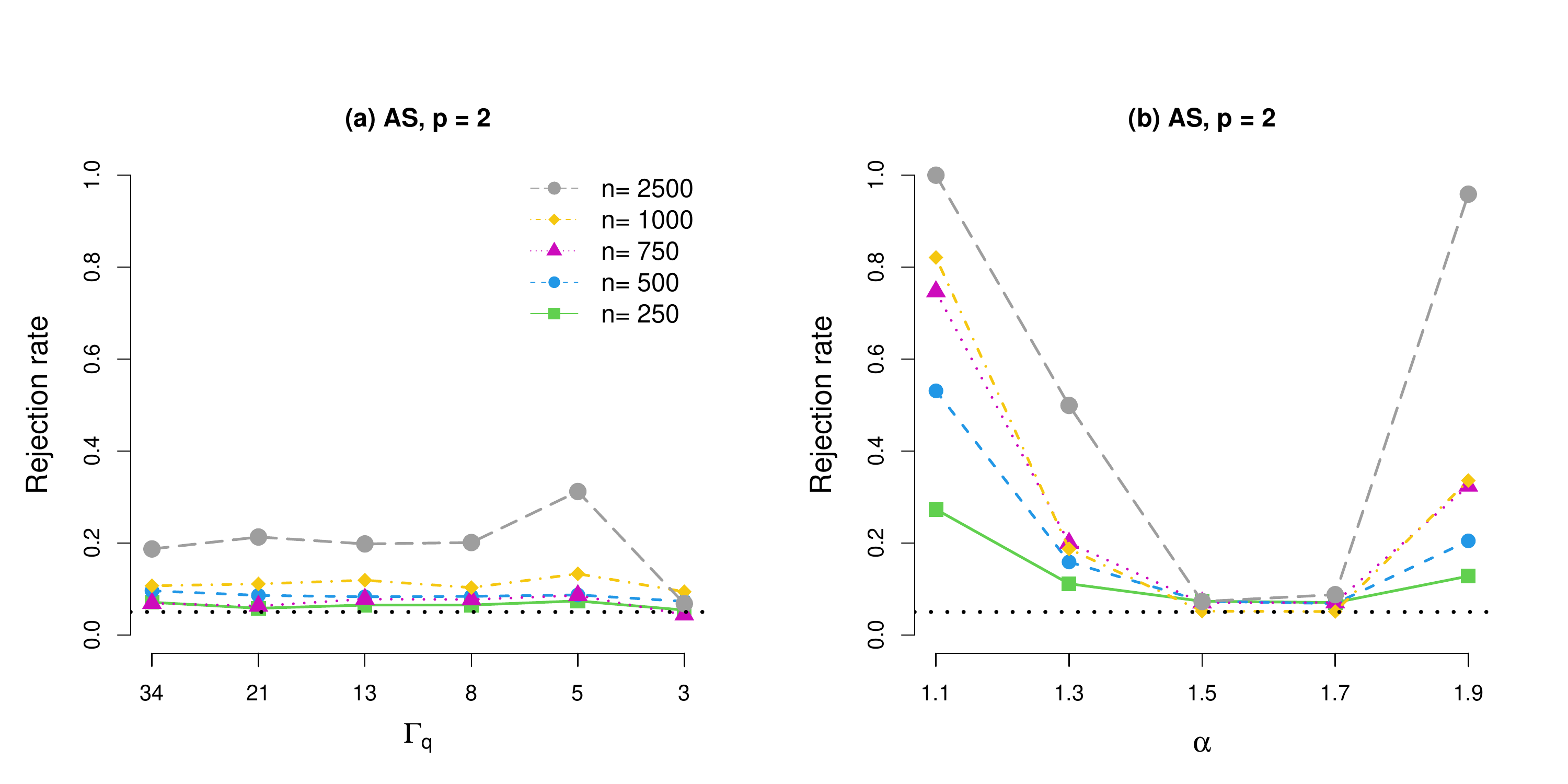}
    \caption{Power functions of the test when the null hypothesis is an AS$_2(\0, \Gamma_0, \alpha_0)$ distribution, with $\alpha_0 = 1.5$ and $\Gamma_0(\cdot) = (1/3) \sum_{k = 1}^3 \mathds{1}_{\bs_k}(\cdot)$, where $\bs_k = (\cos(2\pi k/3), \sin(2\pi k/3))^\top$, for $k \in \{1, 2, 3\}$, corresponding to Simulation \ref{sim10}. On the left-hand side figure, the indexes in the abscissa correspond to the alternative spectral measures $\Gamma_q(\cdot) = (1/q) \sum_{k = 1}^q \mathds{1}_{\bs_k}(\cdot)$, where $\bs_k = (\cos(2\pi k/q), \sin(2\pi k/q))^\top$, for $k \in \{1, \ldots, q\}$ and $q \in \{34, 21, 13, 8, 5\}$. On the right-hand side, we plotted the power of the test when, now, $\Gamma(\cdot)$ is fixed, but the alternatives hypotheses are for $\alpha \in \{1.1, 1.3, 1.7, 1.9\}$.} \label{Fig2}
\end{figure}

\subsection{Estimated power functions for the composite null hypothesis case} \label{subsec53}

For the next simulations, consider the goodness-of-fit testing problem with composite null hypotheses ${\cal{H}}^{\rm c}_0$ as stated in \eqref{null3c} and alternative hypotheses ${\cal{H}}^{\rm c}_1$ as in \eqref{alt3c}. 

\begin{simul}\label{sim11}
 We generated $n$ observations from an ST$_p(\0,\bI,\balpha, \nu)$ distribution with $p \in \{2,3\}$, a fixed value for $\balpha = (3, 0, \ldots, 0)^\top$, and $\nu \in \{1, 2, 3, 5, 8, 13, \infty\}$. For the null hypothesis, we considered the family of skew-normal distributions. Hence the generated observations belong to the set of alternative hypotheses. Then we calculated the empirical power functions and the results are plotted in Figure \ref{Fig3}. We notice that, as $\nu \to \infty$, the ST$_p(\0,\bI,\balpha, \nu)$ distribution converges to the SN$_p(\0,\bI,\balpha)$ distribution. In particular, this effect is also observed in the power functions, with convergence to the significance level, here set equal to $5 \%$. 
\end{simul}

\begin{figure}[h!]
    \centering
    \includegraphics[scale=0.4]{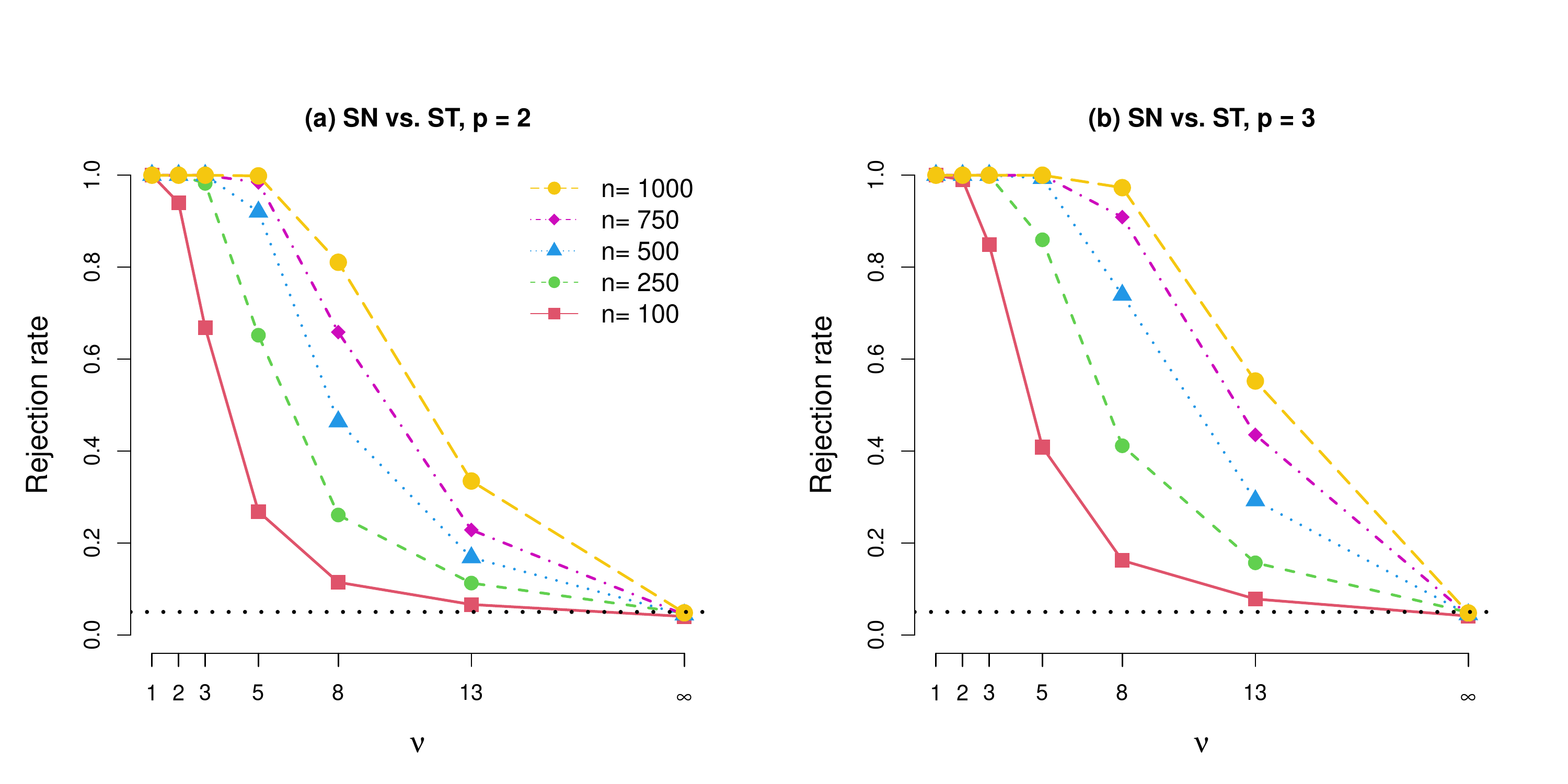}
    \caption{Power of the test when the null hypothesis is assumed to be an SN$_p(\0, \bI, \balpha)$ distribution, for some $\balpha = (\alpha^*, 0, \ldots,0)^\top$ with $\alpha^* \in [0, \infty)$, against the alternative family of ST$_p(\0, \bI, \balpha, \nu)$ distributions, with $p \in \{2,3\}$, $\balpha = (3, 0, \ldots, 0)^\top$ fixed, and $\nu \in \{1, 2, 3, 5, 8, 13, \infty\}$. The horizontal dotted line corresponds to the $5 \%$ significance level.} \label{Fig3}
\end{figure}

\begin{simul}\label{sim12}
 Now consider the family of multivariate sinh-arcsinh distributions introduced by Jones and Pewsey \cite{JP09}. Our aim with this simulation is to compute the power functions by considering the sinh-arcsinh as the alternative hypotheses of our test, and as the null hypothesis we shall consider the five families of distribution presented in Section \ref{sec4}. To define the multivariate sinh-arcsinh distribution, let us consider the univariate transformation 
 \be \label{Sab}
    {\rm S}_{a,b} (z) =  \sinh\{b \sinh^{-1} (z) - a\}, \quad a \in \R,\ b \in \R_+. 
  \ee
 Also, let $(\bm{e},\bm{f}) \in \R^p \times \R^p_+$, with $\bm{e} = (e_1,\ldots, e_p)^\top$ and $\bm{f} = (f_1, \ldots, f_p)^\top$, and $\bZ \sim {\rm N}_p(\0, \bI)$ follow a standard $p$-variate Gaussian distribution. Then, applying \eqref{Sab} component-wise, we say that
 \be
    \bY_{\bm{e}, \bm{f}} = \bm{\rm S}_{-\bm{e}/\bm{f}, 1/\bm{f}} (\bZ) := \left({\rm S}_{-\frac{e_1}{f_1}, \frac{1}{f_1}} (Z_1), \ldots, {\rm S}_{-\frac{e_p}{f_p}, \frac{1}{f_p}} (Z_p)\right)^\top
 \ee
 has a sinh-arcsinh distribution with parameters $(\bm{e},\bm{f})$. We simulated $n$ random samples from $\bY_{\bm{e}, \bm{f}}$ when the parameter $\bm{e} = (e, \ldots, e)^\top$, with $e  \in \{0, 0.1, 0.2, 0.3, 0.4, 0.5\}$, and $\bm{f} = (f, \ldots, f)^\top$, with $f = (e+1)^{-1}$, so that $1/f \in \{1, 1.1, 1.2, 1.3, 1.4, 1.5\}$ and $-e/f \in \{0, -1/11, -1/6, -3/13, -2/7, -1/3\}$. In particular, when $e = 0$, $\bY_{\bm{e}, \bm{f}}$ has a standard Gaussian distribution. Moreover, as $e$ increases, its distribution departs rapidly from the standard Gaussian distribution and turns out to be positively skewed on each axis concerning the origin. For the generation of these random samples, we considered two settings. In the first setting, we take $n \in \{50, 100, 250, 500\}$ and $m = n$ in our tests. While in the second set, we fixed the size of the generated samples, with $n = 100$, and shifted $m$ along the set $\{100, 250, 500, 1000\}$, the sample size of the newly generated data needed for the tests, as discussed in Section \ref{sec2}. Then we calculated the empirical power functions for each test, with the five different families of distributions considered in Section \ref{sec4} to be the designed null hypothesis. The results found for each of the two different settings are plotted in Figures \ref{Fig4} and \ref{Fig5}. The latter figure shows that, if we have a small sample size data (in this particular simulation, with $n = 100$), gradually increasing the value of $m$ from $100$ to $1000$ also slightly increases the power, which is good to know in cases of small sample size; see for instance, the AIS data set from Subsection \ref{subsec61}.\end{simul}

\begin{figure}[h!]
    \centering
    \includegraphics[scale=0.52]{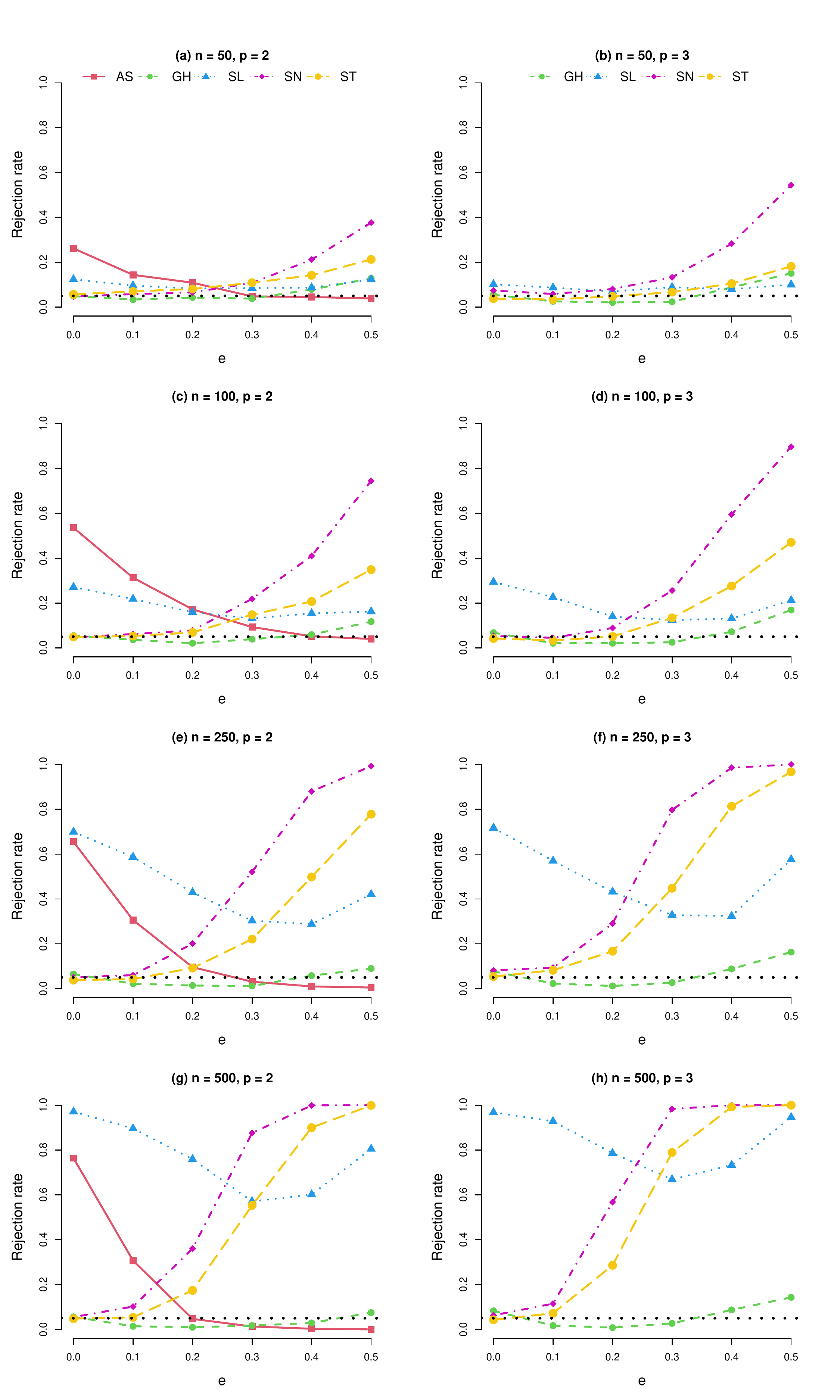}
    \caption{Power of the tests against the alternative family of sinh-arcsinh distributions correspondent to Simulation \ref{sim12}. Here the sample sizes considered are $n \in \{50, 100, 250, 500\}$ and $m = n$.} 
    \label{Fig4}
\end{figure}

\begin{figure}[h!]
    \centering
    \includegraphics[scale=0.52]{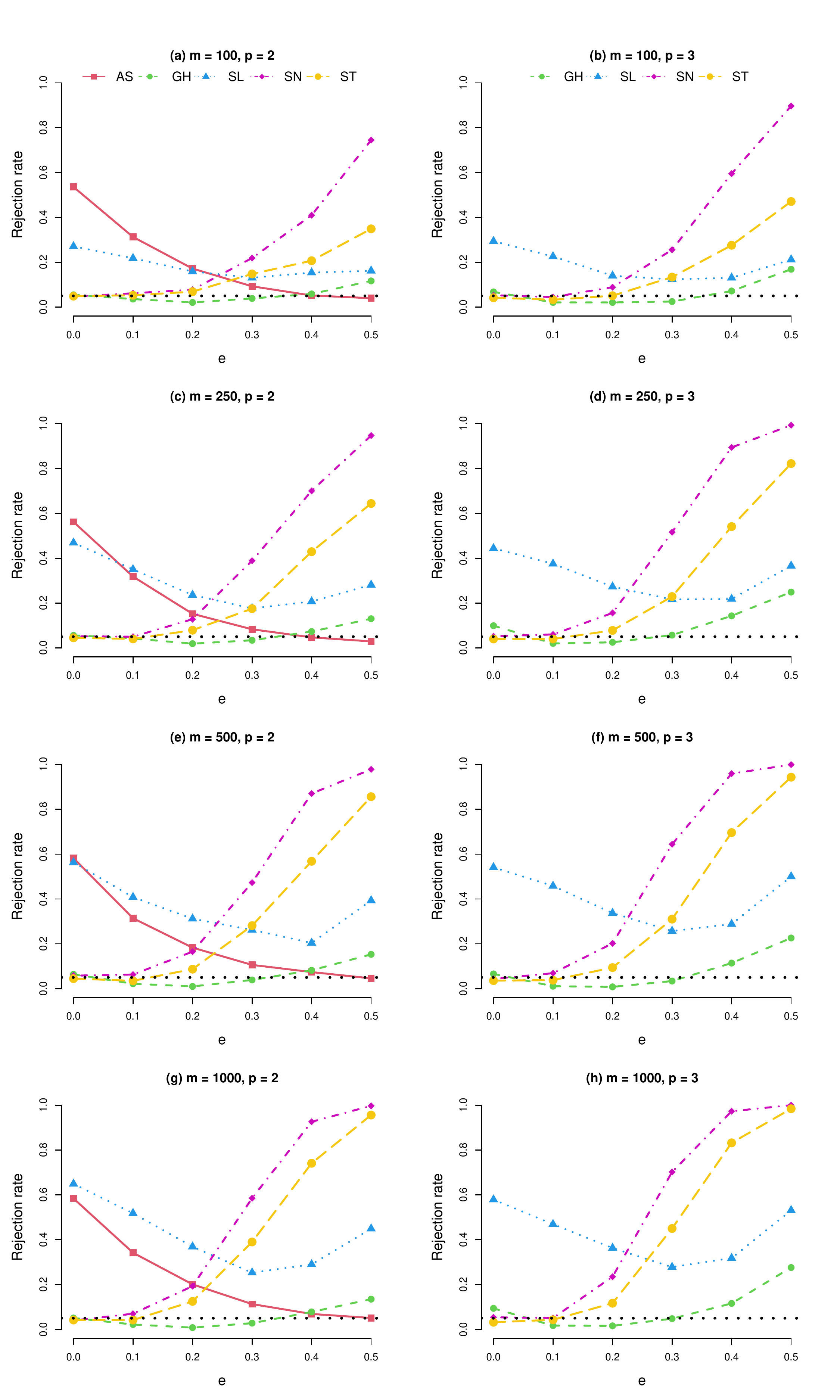}
    \caption{Power of the tests against the alternative family of sinh-arcsinh distributions correspondent to Simulation \ref{sim12} with sample size $n=100$ fixed and $m \in \{100, 250, 500, 1000\}$.} 
    \label{Fig5}
\end{figure}

\subsection{Comparison over competitor tests} \label{subsec54}

Although the major tests proposed in the literature are restricted to the SN family, we dedicate this subsection to discussing and comparing our test to these alternatives. In a retrospective overview, Meintanis and Hl\'avka \cite{MH10} introduced one of the first goodness-of-fit tests for the family of multivariate skew-normal distributions which utilizes the empirical moment-generating function. However, the computational formulas for this test are practically restricted to the two-dimensional case since the case $p > 2$ requires a solution of a differential equation that is difficult to be numerically evaluated. Later on, Balakrishnan et al. \cite{BCS14} proposed a test that is based on the skew-normal's canonical form. Its main advantage relies on the argument that no re-sampling step is needed, saving computational time. However, it is only valid for the skew-normal distribution as the test is based on the fact that the ratios $X_{i,j}/|X_{i,1}|$, for $j \in \{2, \ldots, p\}$, of the observations $\bX_i = (X_{i,1}, \ldots, X_{i,p})^\top$, for $i \in \{1, \ldots, n\}$, are distributed as Cauchy random variables if the $\bX_i$'s follow a canonical skew-normal distribution. More recently, Gonz\'{a}lez-Estrada et al. \cite{GVA22} introduced two randomized tests that, similarly to Balakrishnan et al. \cite{BCS14}'s test, are based on the estimated canonical form of the SN distribution. The first test (W) applies the principle of a generalization of the Shapiro-Wilk test after the sample is transformed into approximately multivariate standard normal observations, whereas the second test (S) relies on a closure property of the sum of univariate independent skew-normal and normal random variables. Jim\'{e}nez-Gamero and Kim \cite{JK15} proposed a pair of re-sampling schemes, one of which is the parametric bootstrap (PB in the paper) which is what we also use in our tests. 

In the next simulation, we replicate one of the original simulation studies from Balakrishnan et al. \cite{BCS14} and compare the estimated powers of our test with the ones reported by these authors and the ones given in Gonz\'{a}lez-Estrada et al. \cite{GVA22}. 

\begin{simul}
  This study takes into account the $3$-variate skew-normal distribution for the composite null hypothesis, and the ST$_3(\bxi, \bOmega, \balpha, \nu)$ distribution in the alternative set,  with $\bxi$, $\bOmega$, and $\balpha$ defined as follows   
  \be \label{param.BCS}
    \bxi =    \begin{pmatrix}
              1\\ 2\\ 3
            \end{pmatrix}, \qquad
    \bOmega = \begin{pmatrix}
              1 & 1 & 1\\ 1 & 2.5 & 1\\ 1 & 1 & 5
            \end{pmatrix}, \qquad
    \balpha = \begin{pmatrix}
              1 \\ -2 \\ 3
            \end{pmatrix}.   
  \ee      
  The degrees of freedom $\nu$ are taken in the set $\{1, 2, 3, 5, 10\}$. We generated ten rounds of replications of our test with different seeds when the sample size is $n = 100$ and $m = 1000$.  The rejection rates are presented in Table \ref{tab.comp}. For comparison reasons, we transcribed the values of the powers reported for this case in Tables 2 and 5 from Balakrishnan et al. \cite{BCS14} and Gonz\'alez-Estrada et al. \cite{GVA22}. As we can see, all ten rounds of tests have shown higher powers than Balakrishnan et al. \cite{BCS14}'s test. In comparison with Gonz\'alez-Estrada et al. \cite{GVA22}'s test, except for $\nu = 10$, our test also has shown higher powers. This shows that, in particular, when the SN distribution is being tested against the ST distribution, our tests are equivalent or even better in terms of powers than the competitor tests presented in the literature.
\end{simul}
 
\begin{table}[!h]
    \centering
    \caption{Simulated values of power from $1000$ replications when the sample size is $n=100$ and the data originates from a ST$_3(\bxi, \bOmega, \balpha)$ distribution with $\bxi$, $\bOmega$, and $\balpha$ defined as in \eqref{param.BCS} and the composite null hypothesis is assumed to be in the skew-normal family. Values reported in boldface characters indicate the estimated powers that are lower than Gonz\'alez-Estrada et al. \cite{GVA22}'s test.}
    \label{tab.comp}
    \begin{tabular}{llllll}
      \hline
        $\nu$ (degrees of freedom) & 1 & 2 & 3 & 5 & 10\\
      \hline 
        Balakrishnan et al. \cite{BCS14}' test & 0.864 & 0.475 & 0.277 & 0.193 & 0.165 \\
        Gonz\'{a}lez-Estrada et al. \cite{GVA22}'s test & 1.000 & - & 0.979 & 0.758 & 0.299\\[2mm] \hline
        Our test - Round 1 & 1.000 & 1.000 & 0.993 & 0.838 & 0.308\\
        Our test - Round 2 & 1.000 & 1.000 & 0.988 & 0.797 & \textbf{0.247}\\
        Our test - Round 3 & 1.000 & 1.000 & 0.992 & 0.808 & \textbf{0.273}\\
        Our test - Round 4 & 1.000 & 0.999 & 0.989 & 0.801 & \textbf{0.257}\\
        Our test - Round 5 & 1.000 & 1.000 & 0.994 & 0.785 & \textbf{0.281}\\
        Our test - Round 6 & 1.000 & 1.000 & 0.990 & 0.816 & 0.321\\
        Our test - Round 7 & 1.000 & 1.000 & 0.994 & 0.791 & \textbf{0.282}\\
        Our test - Round 8 & 1.000 & 1.000 & 0.991 & 0.778 & \textbf{0.265}\\
        Our test - Round 9 & 1.000 & 1.000 & 0.995 & 0.791 & \textbf{0.238}\\
        Our test - Round 10 & 1.000 & 1.000 & 0.993 & 0.806 & \textbf{0.269}\\
      \hline  
    \end{tabular} 
\end{table}


\section{Data applications} \label{sec6}  

This section considers some examples with real-data samples previously presented in the literature. We apply and discuss the results of the goodness-of-fit tests proposed in the present paper. We considered a $5 \%$ confidence level for each test that we performed. The p-values that indicate rejection of the null hypothesis are shown in boldface characters in the tables below. Additionally, to obtain more power for each test, we set $m = \max\{n, 1000\}$ to be in accordance with the results presented in Simulation \ref{sim12}.

\subsection{AIS data set}\label{subsec61}

The Australian Institute of Sport (AIS) data set is one of the classical examples presented by Azzalini and Capitanio \cite{AC99} to illustrate the fitting of a skew-normal distribution. The data consists of biomedical measurements on 100 female and 102 male athletes collected at the Australian Institute of Sport, including body mass index (BMI), body fat percentage (BFP), the sum of skin folds (SSF), and lean body mass (LBM), among others, and it can be retrieved through the \texttt{sn} \cite{A22} package in \textsf{R} \cite{R22}. These four mentioned indexes were also recently used by Balakrishnan et al. \cite{BCS14} and Gonz\'{a}lez-Estrada et al. \cite{GVA22} for testing the goodness-of-fit of the skew-normal distribution. Here, in addition to including tests for the skew-normal family, we also include the tests for asymmetric $\alpha$-stable (two-dimensional case only), Tukey $g$-and-$h$, skew-Laplace, and skew-t distributions. 

We applied our tests on the two-dimensional and four-dimensional data with the athletes segregated by gender, female and male, and we obtained the estimated p-values shown in Table \ref{tabAIS}. We observe that the only test that failed to reject the null hypothesis in the four-dimensional case, for both female and male athletes, was the one with the GH distribution. As for the SL and ST distribution, the tests suggest the rejection of the null hypotheses only for the data on female athletes. Moreover, the test leads to the conclusion in favor of the SN distribution for the data on female athletes, while for the data on male athletes, the test suggests the rejection of the SN distribution. These two tests are, therefore, in line with the conclusions presented by Balakrishnan et al. \cite{BCS14} and Gonz\'{a}lez-Estrada et al. \cite{GVA22}. For the pairwise two-dimensional case, we observe that most of the tests with the AS distribution suggest rejection of the null hypothesis, with 10 out of 12 pairs of data showing p-values lower than 0.05. 

\medskip

\begin{table}[h!]
    \centering
    \caption{Estimated p-values of the tests for the four-dimensional and pairwise two-dimensional AIS data set.} \label{tabAIS}
    \begin{tabular}{lllll}
        \hline
        \multicolumn{5}{c}{\textbf{Four-dimensional AIS data}}\\
        \hline
         & GH & SL & SN & ST\\
        Female athletes 
        & 0.715 & \textbf{0.004} & 0.068 & \textbf{0.022}\\[2mm]  
        Male athletes 
        & 0.175 & 0.503 & \textbf{0.013} & 0.247\\[3mm]
    \end{tabular}
    \begin{tabular}{lllllll}
      \hline
        \multicolumn{7}{c}{\textbf{Pairwise two-dimensional AIS data}} \\      
      \hline
        & & AS & GH & SL & SN & ST \\
        \multirow{6}{*}{\begin{turn}{90}Female athletes\end{turn}} 
        & BMI \& BFP & \textbf{0.002} & 0.074 & \textbf{0.003} & 0.234 & 0.201 \\
        & BMI \& SSF & \textbf{0.000} & \textbf{0.000} & \textbf{0.008} & \textbf{0.044} & \textbf{0.010} \\
        & BMI \& LBM & \textbf{0.000} & 0.618 & 0.203 & 0.102 & 0.236 \\
        & BFP \& SSF & \textbf{0.000} & 0.287 & \textbf{0.001} & \textbf{0.049} & \textbf{0.037} \\
        & BFP \& LBM & \textbf{0.000} & \textbf{0.000} & \textbf{0.010} & 0.115 & 0.087 \\
        & SSF \& LBM & 0.271 & 0.105 & \textbf{0.012} & 0.162 & 0.123 \\[2mm]
        \hline
        \multirow{6}{*}{\begin{turn}{90}Male athletes\end{turn}} 
        & BMI \& BFP & \textbf{0.000} & 0.258 & \textbf{0.033} & \textbf{0.001} & 0.648 \\
        & BMI \& SSF & 0.257 & 0.058 & 0.209 & \textbf{0.003} & 0.778 \\
        & BMI \& LBM & \textbf{0.001} & 0.157 & \textbf{0.038} & 0.300 & 0.486 \\
        & BFP \& SSF & \textbf{0.022} & \textbf{0.018} & 0.213 & \textbf{0.002} & 0.204 \\
        & BFP \& LBM & \textbf{0.000} & 0.281 & 0.052 & \textbf{0.032} & 0.243 \\
        & SSF \& LBM & \textbf{0.041} & \textbf{0.001} & 0.118 & 0.199 &  0.472\\
      \hline
    \end{tabular}
\end{table}

\subsection{BMI of Australian twin sample biometric data} \label{subsec62}

Nowadays it is clear from a statistical perspective that the BMI's population distribution is not symmetric, usually showing skewness to the right towards a higher ratio of weight to height (see Nuttall \cite{N15}). By considering the BMI observations of the AIS data discussed in Subsection \ref{subsec61}, Marchenko and Genton \cite{MG10} presented strong evidence that the skewness parameter is different from zero. It has also been pointed out in the literature (see, e.g., Tran et al \cite{TWA18} and Tsang et al. \cite{TDDT18}) that the skew-t distribution is reasonably competitive when describing unimodal BMI data. So, as our second application, we consider the observations of BMI of monozygotic (MZ) twins retrieved from the \texttt{twinData} set, available in the \texttt{OpenMx} \cite{B22} package in \textsf{R} \cite{R22}. The reason why we decided to use this data set, instead of the AIS data, is because it has more observations and they are more homogeneous. In our analysis, we consider individuals of all ages, separated by gender, with 1171 pairs of females and 532 pairs of males, and we only removed the pairs of twins that showed missing BMI data. 

We fitted the two-dimensional vectors of observed BMI to the five distributions introduced in Section \ref{sec4} and applied our goodness-of-fit tests. The estimated p-values are presented in Table \ref{tabtwinMZ}. For the $5 \%$ confidence level, the tests rejected the SN and ST distributions for both female and male MZ twins. Nevertheless, it is interesting to observe that the p-value of the ST test is significantly higher than the one obtained by the SN test, which corroborates the claims found in the literature that the ST distribution is reasonably better. As for the AS and SL distributions, the tests showed ambiguous results for the two genders, rejecting the two distributions in the male case and showing a relatively high p-value for the female case. Lastly, the test with the GH distribution did not show enough evidence to reject the null hypothesis for both data on females and males.

\begin{table}[h!] 
    \centering
    \caption{Estimated p-values of the test correspondent to the two-dimensional BMI of Australian monozygotic twin sample data sets.} \label{tabtwinMZ}
    \begin{tabular}{llllll}
     \hline  
      & AS & GH & SL & SN & ST \\
      Female MZ twins 
      & 0.458 & 0.290 & 0.288 & \textbf{0.000} & \textbf{0.042} \\[2mm]
      Male MZ twins   
      & \textbf{0.000} & 0.216 & \textbf{0.000} & \textbf{0.012} & \textbf{0.046} \\
     \hline
    \end{tabular}
\end{table}

\subsection{Wind speed data}

As a third and final example, we tested the \emph{wind speed} data set presented in Azzalini and Genton \cite{AG08}), consisting of 278 observations of hourly average wind speed measurements from February 25 to November 30, 2003, recorded at midnight and collected at three meteorological towers: Goodnoe Hills (gh), Kennewick (kw), and Vansycle (vs), located along the Columbia Gorge and the Oregon–Washington border in the US Pacific Northwest. Azzalini and Genton \cite{AG08} proposed the fitting of the data by using an i.i.d. skew-t three-dimensional model, claiming that it ``brings significant improvements over the normal distribution''. This same data set was also used by Arslan \cite{A10} to illustrate the fitting of the skew-Laplace distribution. The author considered the two-dimensional vectors of wind speed recorded at the towers (gh, kw) and (vs, gh), arguing that the data were satisfactorily fitted to the scatterplots by the skew-Laplace distribution and that it captured the skewness and the apparent heavy tailedness.

We run our goodness-of-fit tests on the tri-dimensional wind speed data set to verify if any of the skewed models introduced in Section \ref{sec4} is inappropriate. To get additional information, we also applied the same tests to the pairwise two-dimensional data sets, now including the $\alpha$-stable distribution. The estimated p-values are presented in Table \ref{tabWS}. Considering the $5 \%$ level of significance, in the tri-dimensional case, only the test for the GH distribution did not show enough evidence for rejecting the null hypothesis, whereas all the other tests, namely, for SL, SN, and ST distributions, presented a p-value lower than $0.05$, thus suggesting the rejection of these three distributions. In the pairwise two-dimensional case, most of the tests suggest rejection of the null hypothesis. This conclusion might be because the data shows signs of bi-modality and perhaps a mixture of distributions is more appropriate to model this data set.
 
\begin{table}[h!]
    \centering
    \caption{Estimated p-values of the goodness-of-fit tests correspondent to the tri-dimensional and pairwise two-dimensional wind speed data set.} \label{tabWS} 
    \begin{tabular}{cccc}
      \hline
        \multicolumn{4}{c}{\textbf{Three-dimensional data}} \\
       \hline 
        GH & SL & SN & ST\\
        0.938 & \textbf{0.000} & \textbf{0.000} & \textbf{0.000}\\[3mm]
    \end{tabular}\\
    \begin{tabular}{lccccc}
      \hline
        \multicolumn{6}{c}{\textbf{Pairwise two-dimensional data}} \\
      \hline  
        & AS & GH & SL & SN & ST\\
        gh \& kw &
        \textbf{0.006} & \textbf{0.000} & \textbf{0.000} & \textbf{0.002} & \textbf{0.006} \\
        gh \& vs &
        \textbf{0.001} & 0.067 & \textbf{0.000} & \textbf{0.000} & \textbf{0.004} \\
        kw \& vs &
        0.397 & \textbf{0.001} & \textbf{0.000} & \textbf{0.000} & \textbf{0.000} \\
      \hline
    \end{tabular}  
\end{table}


\section{Conclusion} \label{sec7}  

In this paper, we proposed a goodness-of-fit test for several types of multivariate skewed distributions. On the one hand, the major advantage of the technique addressed in our work resides in the fact that it is flexible and can be applied to any multivariate parametric family of distributions, provided that a reasonable method of estimation of its parameters is available and that the generation of new replicates is feasible. On the other hand, in terms of computational cost, the implementation is highly demanding since the parametric bootstrap step requires an extra cycle of re-sampling within each Monte Carlo run. 

While the need for such nested re-sampling is shared by most goodness-of-fit tests available in the literature,  this drawback can be easily circumvented with the use of a parallel algorithm, since the parametric bootstrap does not require any sequential procedures, and with the use of the warp-speed bootstrap method of Giacomini et al. \cite{GPW13}. An important fact to be mentioned is that all tests were run with the help of an Intel Xeon Gold 6230R CPU, of which 100 out of its 104 threads have been intensively used to accelerate even more the completion of the simulations. We demonstrated its effectiveness through five families of multivariate distributions, namely, the multivariate skew-normal, skew-t, asymmetric skew-Laplace, skew $\alpha$-stable, and Tukey $g$-and-$h$ (for most of which there are no available tests), by utilizing the corresponding canonical forms whenever possible.

As the simulations in Subsection \ref{subsec51} show (see Table \ref{tab1}), the estimated sizes of the test are reasonable and consistent for all five families. Similarly, the simulations presented in Subsections \ref{subsec52} and \ref{subsec53} show that the tests have enough power to detect and reject alternative hypotheses. Compared to the alternative options of tests introduced in the literature, as presented in Subsection \ref{subsec54}, for the particular case when testing under the composite null hypothesis of an SN distribution, our test has also shown to be better in terms of power when testing against the alternative ST distribution. 

The effectiveness of our tests has also been illustrated with real data examples in Section \ref{sec6}, showing its applicability and usefulness when applied to biological and natural events observed, respectively, in our daily lives and our environment. In closing we wish to remind the reader that our test allows a certain flexibility concerning the actual kernel $\Psi$ used; refer to the last paragraph of Section \ref{sec2}. In this connection, it would be interesting to investigate the effect that this choice has on the finite-sample properties of our test. More work is needed in this direction.     

\section*{Acknowledgements}
\noindent
This research was supported by the King Abdullah University of Science and Technology (KAUST). 

\section*{Competing interests}
\noindent
The authors have no financial or proprietary interests in any material discussed in this article.

\bibliographystyle{acm}

{\footnotesize
\bibliography{references}}

\end{document}